\documentclass[10pt,journal,a4paper,final,oneside,twocolumn]{IEEEtran}

\usepackage{times} 
\usepackage{graphicx}
\usepackage{cite}
\usepackage{citesort}
\usepackage{color}
\usepackage{psfrag}
\usepackage{subfigure}
\usepackage{amssymb}
\usepackage{lmodern}
\usepackage{epsfig}
\usepackage{pifont}
\usepackage{amsmath}

\usepackage{cases}
\usepackage{array}
\usepackage{multicol}
\usepackage{multirow}
\usepackage[T1]{fontenc}
\usepackage{comment}
\usepackage{amsthm}
\usepackage{indentfirst}
\usepackage{cases}
\usepackage[ruled, lined, linesnumbered, commentsnumbered, longend]{algorithm2e}
\SetKwInOut{KwIn}{Input}
\SetKwInOut{KwOut}{Output}
\SetKwProg{Init}{Initialization}{:}{end}
\SetKwProg{QT}{[QT Iteration]}{}{end}
\SetKwProg{DT}{[DT Iteration]}{}{end}
\SetKwProg{LT}{[Bisection Iteration of $\eta_l$]}{}{end}
\graphicspath{{figures/}}
\theoremstyle{remark}
\newtheorem{remark}{Remark}

\newtheorem{lemma}{Lemma}

\makeatletter
\def\ifundefined{\@ifundefined}
\makeatother 

\begin{document}

\title{Spectral-Efficiency of Cell-Free Massive MIMO with Multicarrier-Division Duplex}

\author{Bohan Li, Lie-Liang Yang, {\em Fellow,
    IEEE}, Robert G. Maunder, {\em Senior Member, IEEE}, Songlin Sun, {\em Senior Member, IEEE}, Pei Xiao, {\em Senior Member, IEEE} \thanks{B. Li, L.-L. Yang and R. Maunder are with the School of Electronics
    and Computer Science, University of Southampton, SO17 1BJ,
    UK. (E-mail:
    bl2n18, lly, rm@ecs.soton.ac.uk,~http://www-mobile.ecs.soton.ac.uk/lly). S. Sun is with the School of Information and Communication Engineering, Beijing University of Posts and Telecommunications (BUPT). P. Xiao is with 5GIC \& 6GIC,
Institute for Communication Systems (ICS), University of Surrey, Guildford
GU2 7XH, UK. (Email: p.xiao@surrey.ac.uk)
  The project
    was supported in part by the EPSRC, UK, under Project EP/P034284/1 and EP/P03456X/1,
    and in part by the Innovate UK project.}}

\maketitle

\begin{abstract}
A multicarrier-division duplex (MDD)-based cell-free (CF) scheme, namely MDD-CF, is proposed, which enables downlink (DL) data and uplink (UL) data or pilots to be concurrently transmitted on mutually orthogonal subcarriers in distributed CF massive MIMO (mMIMO) systems. To demonstrate the advantages of MDD-CF, we firstly study the spectral-efficiency (SE) performance in terms of one coherence interval (CT) associated with access point (AP)-selection, power- and subcarrier-allocation. Since the formulated SE optimization is a mixed-integer non-convex problem that is NP-hard to solve, we leverage the inherent association between involved variables to transform it into a continuous-integer convex-concave problem. Then, a quadratic transform (QT)-assisted iterative algorithm is proposed to achieve SE maximization. Next, we extend our study to the case of one radio frame consisting of several CT intervals. In this regard, a novel two-phase CT interval (TPCT) scheme is designed to not only improve the SE in radio frame but also provide consistent data transmissions over fast time-varying channels. Correspondingly, to facilitate the optimization, we propose a two-step iterative algorithm by building the connections between two phases in TPCT through an iteration factor. Simulation results show that, MDD-CF can significantly outperform in-band full duplex (IBFD)-CF due to the efficient interference management. Furthermore, compared with time-division duplex (TDD)-CF, MDD-CF is more robust to high-mobility scenarios and achieves better SE performance. 
\end{abstract}


\IEEEpeerreviewmaketitle

\section{Introduction}\label{section:MDDCF:intro}

As one of the promising techniques integrating the advantages of cloud radio access network (C-RAN), massive multiple-input multiple-output (mMIMO) and coordinated multipoint (COMP), cell-free (CF)-mMIMO has attracted the growing attention from both academia and industry in recent years \cite{demir2021foundations}. The inter-cell interference, which imposes the main limit on the performance of the cell-edge users in cellular networks, is no longer intractable in the CF-mMIMO systems, as the result of no cell boundaries and user-centric operation, and hence leading to the possible seamless coverage \cite{ngo2017cell}. Furthermore, CF-mMIMO has been deemed as the most promising technique to support massive machine type communications in the future Internet-of-things (IoT) networks, owing to its better massive access performance and more flexible AP cooperation \cite{ke2020massive,ye2021deep}.

As it is evolved from the collocated mMIMO, CF-mMIMO systems still consider time-division duplex (TDD) as the dominant duplexing mode \cite{sanguinetti2019toward,interdonato2019ubiquitous}. However, TDD has its limitation, especially, when implemented in fast time-varying environment. First, in TDD, the guard period (GP) between DL and UL transmissions is indispensable. In this regard, when communication channels change quickly in high-mobility scenarios, UL training has to be frequently performed due to the shorter CT interval, and hence the portion of time for data transmission decreases \cite{li2021multicarrier}. Second, the successive transmission of DL and UL represses the explosive demand of asymmetric communications \cite{kim2020dynamic}. For instance, several users having high demand on DL data may have to wait for a single user with a low demand on UL data, which in turn causes the low utilization efficiency of the time-frequency resources. To mitigate the above-mentioned problems, in-band full-duplex (IBFD) has been envisaged as one of the candidate techniques in beyond 5G and 6G systems, owing to its potential to double the spectral-efficiency (SE), when compared with the half-duplex (HD) systems \cite{kolodziej2019band}. Nevertheless, the self-interference (SI) problem stands in the way of the practical implementation of IBFD, leading to the fact that IBFD can only outperform HD when powerful and efficient SI cancellation (SIC) approaches operated in propagation-, analog- and digital-domain are available, which is however at the expense of high complexity and system resources \cite{goyal2015full,shende2017half,mirza2018performance}. 

Recently, multicarrier-division duplex (MDD) has been proposed as a potential alternative for shift from HD to IBFD \cite{yang2009multicarrier}, which allows systems to render the FD operations in the same time slot and the same frequency band but different subcarriers without the cost of GP. Although MDD still suffers from SI due to the FD operation on the same time-frequency resources, the requirement for the digital-domain SI is largely relaxed with the aid of the FFT operation implemented in the built-in analog-to-digital converter (ADC) of receiver chains \cite{li2020self}. To show the merits of MDD, the comparison of MDD, HD and IBFD has been comprehensively studied in \cite{li2021resource,li2021multicarrier}, when potentially practical scenarios are considered.

\subsection{Motivation}

To date, there have been a lot of works on the TDD-based CF (TDD-CF) \cite{ngo2017cell,du2021cell,interdonato2019ubiquitous,bjornson2020scalable}, which demonstrated the advantages of CF systems over the traditional small-cell systems in terms of throughput and reliability. In order to further improve the performance of TDD-CF, it is intuitively to integrate IBFD into the CF systems, namely IBFD-CF, thereby fully exploiting the resources in time- and frequency-domain compared to its HD counterparts. However, as analyzed in \cite{da2021full}, the inter-AP interference (IAI) and inter-mobile station (MS) interference (IMI) largely limit the performance of IBFD-based multi-AP or/and multi-MS systems.

To elaborate this point further, let us consider a basic example of IAI, as depicted in Fig. \ref{figure-MDDCF-SICdia}. Since in the IBFD-CF systems, the IAI coexists with the desired UL signal in the same resource element, the cancellation of IAI is similar to that of SI. In practice, to avoid large quantization noise and loss of effective ADC bits, the power of the analog-domain signal input to ADC should not exceed the noise floor by the dynamic range of the ADC. In other words, for a 12-bit ADC, which has approximately 42 dB of dynamic range \cite{bharadia2013full}, the maximum signal power input to the ADC is -52 dBm. According to the outdoor microcell PL model presented in \cite{3gpp2017further}, the large-scale fading can lead to over 60 dB attenuation, when two APs are separated by a distance of more than 10 meters. As a result, the IBFD-based AP in CF systems needs to provide 32 dB and 42 dB IAI suppression in propagation/analog-domain and digital-domain, respectively. 

\begin{figure}
\centering
\includegraphics[width=0.8\linewidth]{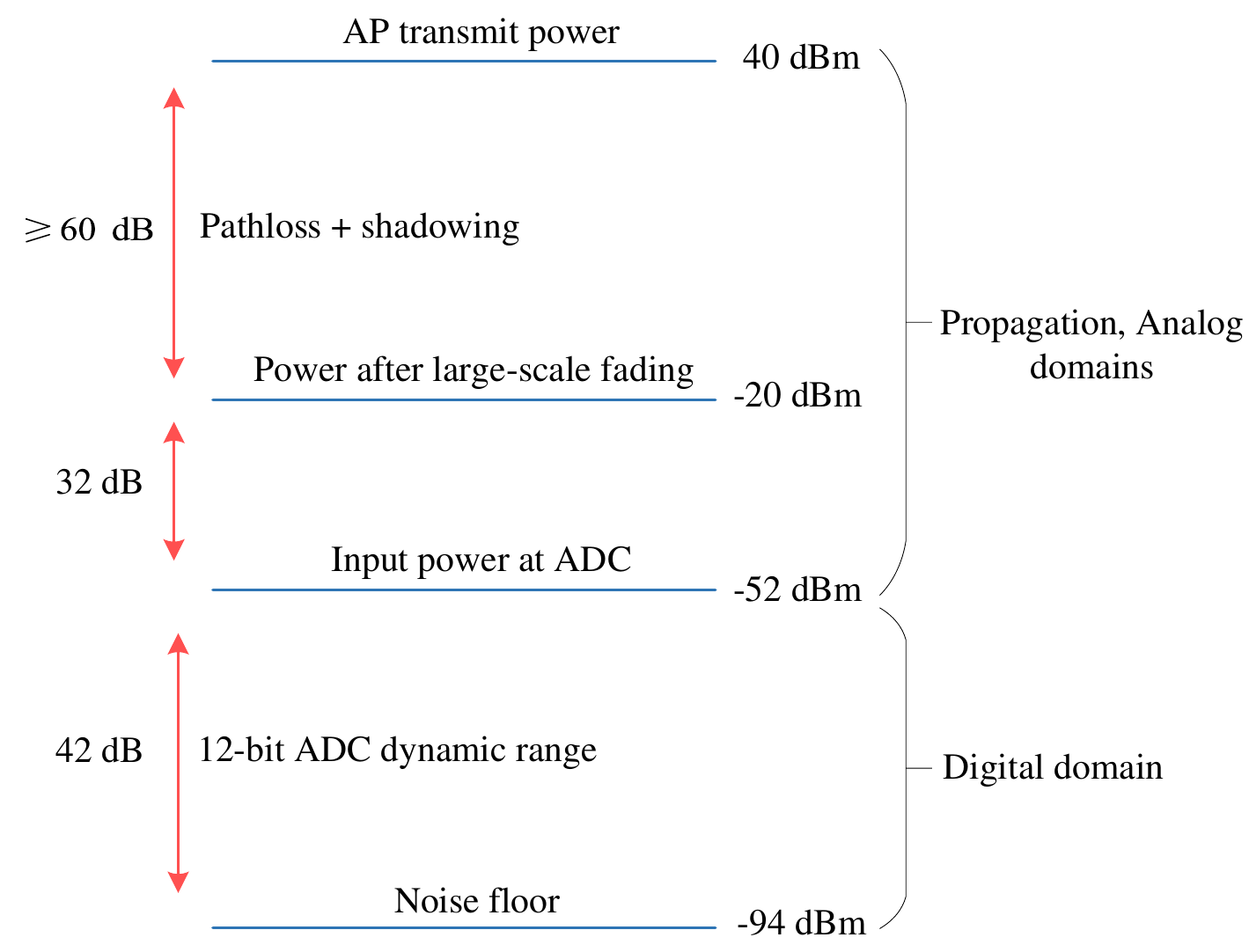}
\caption{An example to illustrate the IAI at receiver. }
\label{figure-MDDCF-SICdia}
\end{figure}

In fact, in comparison with the propagation/analog-domain IAI mitigation, the digital-domain IAI suppression is much harder to be achieved in IBFD-CF. This is because the former can be obtained with the aid of the power control and some passive approaches, such as path-loss, cross-polarization and antenna-directionality \cite{everett2014passive,sabharwal2014band}, while the latter, such as minimum mean square error (MMSE) and successive interference cancellation, are usually not efficient for mitigating the interference generated by a large number of APs from various directions in cell-based systems. For instance, in \cite{nguyen2020spectral}, authors studied the IBFD-CF in centralized CF-mMIMO systems, which can only provide a small amount of digital-domain IAI mitigation by using the coordinated precoding and successive interference cancellation at APs' receiver. In this case, the AP have to cut down its transmit power to decrease the influence of IAI on neighboring APs, which may lower the system performance. Noticeably, compared with IAI, the IMI problem is much more difficult to handle, as the baseband processors in MSs are not as powerful as those in APs. However, since the transmit power of MS is relatively small, the IMI can be largely mitigated in analog domain by power allocation and user scheduling \cite{kim2020dynamic}.

Furthermore, there have been some other works on FD-style CF, such as DTDD-based CF and network-assisted FD-based CF \cite{chowdhury2021can,xia2021joint}, which heavily rely on the central processing unit (CPU) to suppress the digital-domain IAI. However, these centralized CF schemes require the precise IAI channel estimation as well as a large capacity fronthaul, leading to the increased system overhead. Moreover, the influence of propagation/analog-domain IAI/IMI, as we elaborated in Fig. \ref{figure-MDDCF-SICdia}, are overlooked in \cite{chowdhury2021can,xia2021joint}. To the authors' best knowledge, despite the better scalability and less dependence on fronthaul \cite{demir2021foundations}, none of the existing papers study the distributed CF-mMIMO systems with FD-style operation, since the suppression of IAI/IMI can be highly intractable when APs work independently without sharing channel information with the CPU.

\subsection{Contributions}

To alleviate the problem of IAI/IMI, but maximize the time-frequency resource usage at the same time, in this paper, we propose an MDD-CF scheme. According to the principles of MDD, DL and UL transmissions take place at the same time but on different subcarriers. Therefore, the IAI/IMI are mutually orthogonal with the desired UL/DL signals in the digital domain, and can be easily removed during reception without any additional system overhead. 
On the other hand, owing to the FD-style operation, MDD allows to concurrently implement DL transmission and UL training, which improves the SE even over fast time-varying channels. To leverage these advantages, however, one of the challenges is how to efficiently allocate the power and time-frequency resources among multiple APs and MSs for simultaneously supporting data transmission or training in two directions. To the best of our knowledge, there is no study in the open literature that has considered this kind of optimization problem. Therefore, against the background, we comprehensively study the AP-selection, power- and subcarrier-allocation issues in the MDD-CF, with the motivation to achieve the optimal SE performance. In summary, the novelties and contributions of this paper can be briefly described as follows:
\begin{itemize}

\item In order to mitigate the interference to enable the FD-style operation in distributed CF-mMIMO, the MDD-CF scheme is proposed, where the effect of SI, IAI and IMI are practically modeled. Additionally, to study the optimization of AP-selection, power- and subcarrier-allocation under the constraint of MSs' quality of service (QoS), two application scenarios, which assume that DL/UL transmissions occur in one CT interval or in one radio frame, are considered. 

\item In the case of DL/UL transmissions in one CT interval, we firstly leverage the inner association between the continuous variables for power-allocation, and binary variables for AP-selection and subcarrier-allocation, to change the mixed-integer optimization into a continuous-integer convex-concave problem. Then, a quadratic transform (QT)-assisted iterative algorithm is proposed to achieve the SE maximization in MDD-CF scheme.

\item In the case of DL/UL transmissions in one radio frame, we consider imperfect channel estimation. A two-phase CT (TPCT) interval is designed for the CF systems operated in the MDD and IBFD modes. Since the two phases are tightly coupled and both of them support simultaneous transmissions in two directions, they lead to a very intricate formulation. Correspondingly, we introduce an iteration factor to build the connection between the two phases, and transform the original problem to a two-step iterative optimization with the aid of the bisection method. 

\item We comprehensively compare MDD-, IBFD- and TDD-CF in distributed CF-mMIMO systems under the practical network settings. Our simulation results demonstrate the superiority of MDD-CF over IBFD-CF, due to more effective suppression of IAI and IMI in the digital domain. Furthermore, the well-designed TPCT interval with the proposed two-step iterative algorithm enables MDD-CF to achieve much higher SE than TDD-CF in high-mobility communication scenarios.
\end{itemize}

Throughout the paper, the following notations are used: $\pmb{A}$,
$\pmb{a}$ and $a$ stand for matrix, vector, and scalar, respectively;
$\mathcal{A}$ and $\left|\mathcal{A}\right|$ represent the set and the cardinality of set, respectively; $(\pmb{a})_{i}$ denotes the $i$-th element of $\pmb{a}$; $\pmb{A}^{(i,:)}$, $\pmb{A}^{(:,j)}$ and $(\pmb{A})_{i,j}$ denote the $i$-th row, the $j$-th column and the $(i,j)$-th element of $\pmb{A}$, respectively; Furthermore, 
$\left|\pmb{A}\right|$, $\pmb{A}^*$, $\pmb{A}^T$,
$\pmb{A}^{-1}$ and $\pmb{A}^H$ represent, respectively, the
determinant, complex conjugate, transpose, inverse and Hermitian transpose of $\pmb{A}$; The Euclidean norm of a vector and the Frobenius norm of a matrix are denoted as $\left\|\cdot\right\|_2$ and $\left\|\cdot\right\|_F$, respectively; $\pmb{I}_{N}$ denotes a $(N\times N)$ identity matrix;
$\mathcal{CN}(\pmb{0},\pmb{A})$ represents the zero-mean complex
Gaussian distribution with covariance matrix $\pmb{A}$; Furthermore,
$\text{Tr}(\cdot)$, $\text{ln}(\cdot)$ and $\mathbb{E}[\cdot]$ denote
the trace, natural logarithmic and expectation operators, respectively.

\section{System Model}\label{sec:MDD_CF:SM}

\begin{figure}
\centering
\includegraphics[width=0.8\linewidth,height=0.5\linewidth]{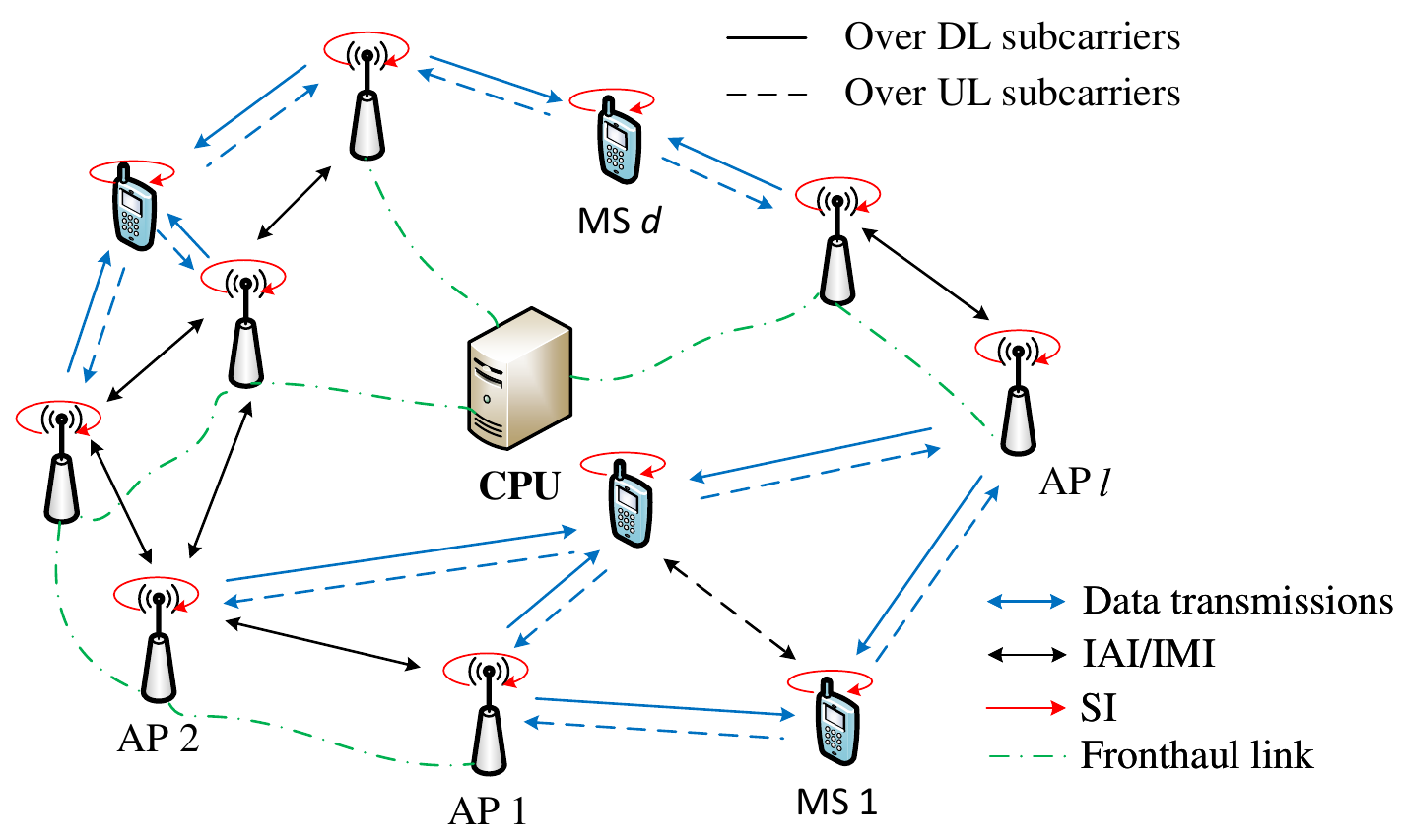}
\caption{Illustration of MDD-CF scheme.}
\label{figure-MDDCF-SM}
\end{figure}

We consider an MDD-CF scheme as shown in Fig. \ref{figure-MDDCF-SM}, where the set $\mathcal{D}=\left\{1,...,d,...,D\right\}$ of single-antenna MSs and the set $\mathcal{L}=\left\{1,...,l,...,L\right\}$ of APs of each with $N$ antennas are operated in the MDD mode that rely on the mutually orthogonal subcarrier sets \cite{li2020self}, namely $\mathcal{M}=\left\{1,...,m,...,M\right\}$ with $\left|\mathcal{M}\right|=M$ and $\mathcal{\bar{M}}=\left\{1,...,\bar{m},...,\bar{M}\right\}$ with $\left|\mathcal{\bar{M}}\right|=\bar{M}$, for DL and UL, respectively. The total number of subcarriers is $M_{\text{sum}} =M+\bar{M}$. Furthermore, we assume that the CF system is operated in a distributed way, where CPU offloads most of the tasks to APs to relieve its computation burden, and only sends coded data to APs for DL transmissions or integrates the received UL data from APs via fronthaul links without any knowledge of channel information. We assume that pilots and DL/UL data are transmitted in a coherence time (CT) interval $T_\text{c}$ in terms of OFDM symbols. As shown in Fig. \ref{figure-MDDCF-Coh}, in the TDD mode, which is deemed as the mainstream duplex mode in mMIMO systems \cite{sanguinetti2019toward}, signal transmissions are performed in sequence, where UL training, guard period (GP) and DL/UL transmissions require $\gamma^{\text{P}}$, $\gamma^{\text{G}}$, $\gamma_{\text{TDD}}^{\text{DL}}$ and $\gamma_{\text{TDD}}^{\text{UL}}$ symbol durations, respectively. Note that, in the TDD mode, the GP between UL and DL is indispensable. By contrast, in both IBFD and MDD modes, after the same training time $\gamma^{\text{P}}$ as that in the TDD mode, DL and UL data transmission time are $T_\text{c}-\gamma^{\text{P}}$, which can be much higher than that in the TDD mode.

\begin{figure}
\centering
\includegraphics[width=0.85\linewidth,height=0.5\linewidth]{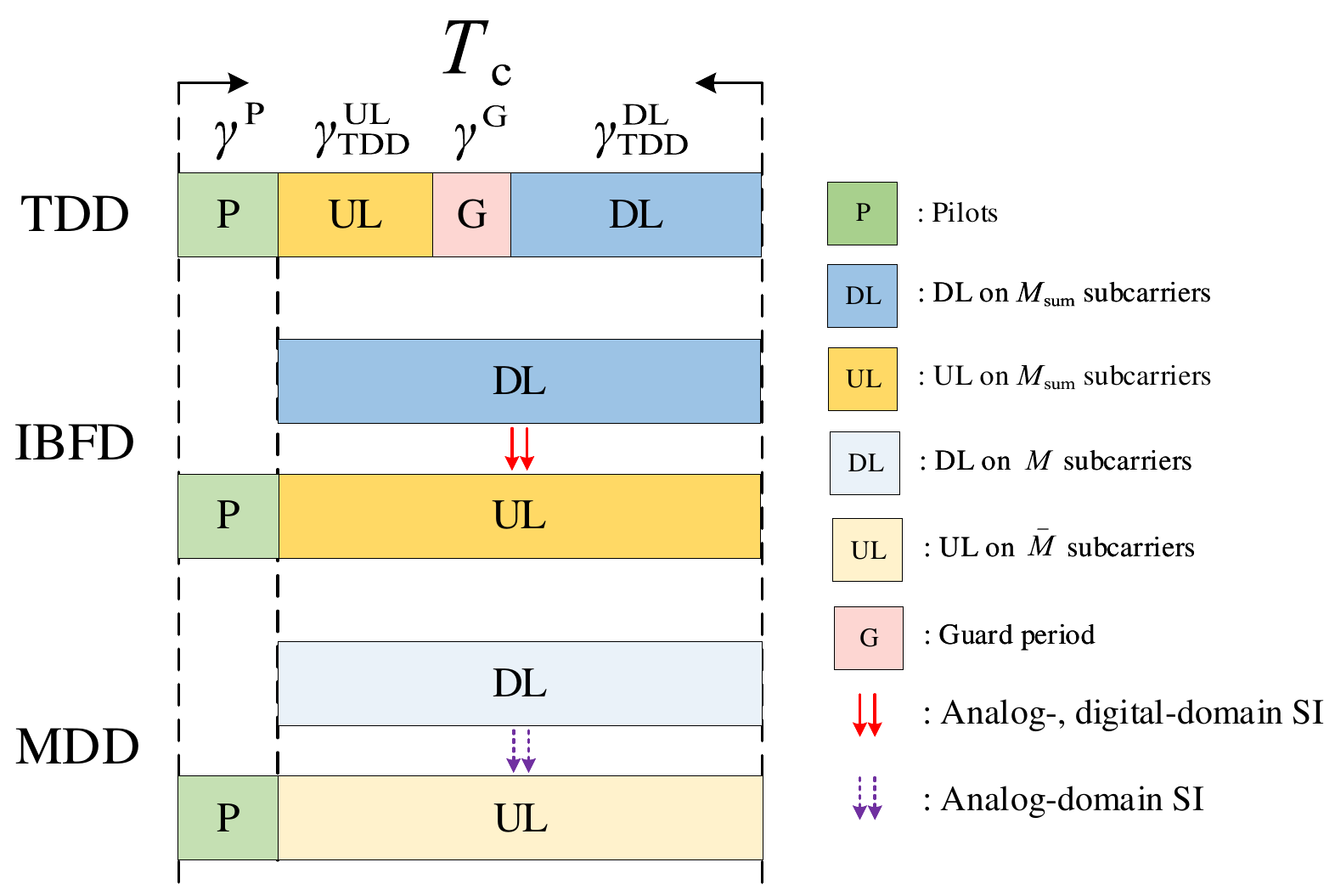}
\caption{Transmission in a coherence time interval of TDD-, IBFD- and MDD- schemes.}
\label{figure-MDDCF-Coh}
\end{figure}

\subsection{Channel Model}
For the convenience of notation, we denote the SI channel at the $l$-th AP and $d$-th MS by $\pmb{H}_{ll} \in \mathbb{C}^{N \times N}$ and $h_{dd}$, respectively. These two SI channels are modeled as
\begin{equation}\label{eq:MDDCF:SICh}
\begin{split}
\left(\pmb{H}_{ll}\right)_{i,j}&=\sqrt{\xi_l^{\text{SI}}} \ \alpha_s,  \\
h_{dd} &= \sqrt{\xi_d^{\text{SI}}} \ \alpha_s,
\end{split}
\end{equation}
where $\alpha_s\sim \mathcal{CN}(0,1)$ is the small-scale fading, $\xi_l^{\text{SI}}$ and $\xi_d^{\text{SI}}\in(0,1]$ denote the residual SI levels at AP and MS receivers, respectively. We assume that $\xi_l^{\text{SI}}<\xi_d^{\text{SI}}$, meaning that the SI at AP receiver can be mitigated to a lower level than that at MS receiver. This is because APs are capable of employing more complicated circuits and higher resource budget for SI suppression than MSs.
\begin{remark}
We consider two assumptions for the SI channels. 1) Since the SI link is relatively short in comparison to the AP-MS link, we assume SI link is of single-path without suffering from the large-scale fading; 2) The residual SI level ($\xi_l^{\text{SI}}/\xi_d^{\text{SI}}$) relies on the SIC capability, which accounts for the propagation/analog- and digital-domain SIC methods\footnote{As we mentioned before, MDD systems can be nearly free from the SI in digital domain due to the FFT operation. By contrast, in order to mitigate the digital-domain SI in IBFD-based systems, the receiver has to estimate the channel between DAC and ADC, and then reconstruct the transmitted signal, which is subsequently subtracted from the received signal \cite{yilan2021integrated,ahmed2015all}. This process incurs large overhead, especially when the system is operated with a large number of subcarriers.}, such as dual-port polarized antenna \cite{debaillie2014analog}, antenna circulator \cite{bharadia2013full} and multi-tap RF canceller \cite{kolodziej2016multitap}.
\end{remark}

Furthermore, we denote the time-domain channel impulse responses (CIRs) of the communication channels between the $d$-th MS and the $n$-th antenna at the $l$-th AP, the IAI channels between the $n$-th antenna at the $l$-th AP and the $n^{\prime}$-th antenna at the $l^{\prime}$-th AP, and the IMI channels between the $d$-th MS and $d^{\prime}$-th MS by $\pmb{g}_{ld}^n \in \mathbb{C}^{U \times 1}$, $\pmb{g}_{ll^{\prime}}^ {nn^{\prime}}\in \mathbb{C}^{U \times 1}$ and $\pmb{g}_{dd^{\prime}} \in \mathbb{C}^{U \times 1}$, respectively, where $U$ is the number of taps of multipath channels. Specifically, the $u$-th tap of these channels can be generally modeled as $\left(\pmb{g}\right)_u=\sqrt{\beta/U} \ \alpha_s$ with $\pmb{g} \in \left\{\pmb{g}_{ld}^n, \pmb{g}_{ll^{\prime}}^ {nn^{\prime}}, \pmb{g}_{dd^{\prime}} \right\}$, where $\beta \in\left\{\beta_{ld}, \beta_{ll^{\prime}}, \beta_{dd^{\prime}}\right\}$ accounts for the large-scale fading of path loss and shadowing. The channels of different taps are assumed to be independent. Additionally, the spatial correlation among the antennas of AP is not considered in this paper, which will be addressed in the future research. 

Given the time-domain CIRs, the frequency-domain channels can be obtained as $\pmb{h}=\pmb{F}\pmb{\Psi}\pmb{g}$ with $\pmb{h} \in \left\{\pmb{h}_{ld}^n, \pmb{h}_{ll^{\prime}}^ {nn^{\prime}}, \pmb{h}_{dd^{\prime}} \right\}$, where $\pmb{F} \in \mathbb{C}^{M_{\text{sum}}\times M_{\text{sum}}}$ is the FFT matrix, $\pmb{\varPsi}\in \mathbb{C}^{{M_{\text{sum}}\times U}}$ is constructed by the first $U$ columns of $\pmb{I}_{M_\text{sum}}$. Moreover, the single DL/UL subcarrier channel can be expressed as $h[m]=\pmb{\phi}_{\text{DL}}^T\pmb{h}$ and $h[\bar{m}]=\pmb{\phi}_{\text{UL}}^T\pmb{h}$, respectively, where $\pmb{\phi}_{\text{DL}}=\pmb{I}_{M_\text{sum}}^{(:,m)}$ and $\pmb{\phi}_{\text{UL}}=\pmb{I}_{M_\text{sum}}^{(:,\bar{m})}$ are the mapping vectors. Note that here $h[m]$ or $h[\bar{m}]$ denotes the point-to-point subcarrier channel, which will be further integrated into the vector and matrix for the AP-AP channel (i.e., $\pmb{H}_{ll^{\prime}}[m \  \text{or} \ \bar{m}] \in \mathbb{C}^{N \times N}$ ), AP-MS channel (i.e., $\pmb{h}_{ld}[m] \in \mathbb{C}^{N\times 1}$ and $\pmb{h}_{ld}[\bar{m}] \in \mathbb{C}^{N\times 1}$), respectively, as shown later in Section \ref{sec:MDD_CF:SM}.B.

\subsection{Downlink Transmission}
Within each CT interval, the data transmitted on the $m$-th DL subcarrier for the $d$-th MS is denoted by $x_d[m]$, which satisfies $\mathbb{E}\left\{\left|x_d[m]\right|^2\right\}=1$. The transmitted signal on the $m$-th DL subcarrier by the $l$-th AP is given by
\begin{equation}\label{eq:MDD_CF:slm}
\pmb{s}_l[m]=\sum_{d \in \mathcal{D}} \lambda_{ld}\mu_{ldm} \sqrt{p_{ldm}}\pmb{f}_{ld}[m] x_d[m],
\end{equation} 
where the binary variable $\lambda_{ld}$ denotes the association relationship between the $l$-th AP and the $d$-th MS, with $\lambda_{ld}=1$ expressing that the $d$-th MS is associated with the $l$-th AP and $\lambda_{ld}=0$, otherwise.  The binary variable $\mu_{ldm}$ explains the operation status of MS $d$ on the $m$-th DL subcarrier, with $\mu_{ldm}=1$ implying that the $m$-th DL subcarrier is activated by the $l$-th AP for DL transmission to MS $d$ and $\mu_{ldm}=0$, otherwise. In \eqref{eq:MDD_CF:slm}, $\pmb{f}_{ld}[m] \in \mathbb{C}^{N \times 1}$ is the  precoding vector with $\left\|\pmb{f}_{ld}[m]\right\|_2^2=1$, and $p_{ldm}$ is the power allocated to the $m$-th subcarrier of the $d$-th MS by the $l$-th AP. The total power budget at the $l$-th AP is expressed as $P_l$, satisfying $\sum_{m \in \mathcal{M}}\sum_{d \in \mathcal{D}} \lambda_{ld}\mu_{ldm} p_{ldm}\leq P_l$.

The signal received from the $m$-th DL subcarrier at the $d$-th MS can be expressed as
\begin{align}\label{eq:MDD_CF:ydm}
y_d[m] = \underbrace{\sum_{l \in \mathcal{L}} \pmb{h}_{ld}^H[m] \pmb{s}_l[m]}_{\text{Desired signal + MUI}} + z^{\text{SI}}_d + z^{\text{IMI}}_d  + n_d,
\end{align}
where $n_d \sim \mathcal{CN}(0,\sigma^2)$ is the additive white Gaussian noise. According to \cite{day2012full2,ng2016power}, the residual interference in digital domain arising from SI and IMI \big(i.e., $z^{\text{SI}}_d$ and $z^{\text{IMI}}_d$ in \eqref{eq:MDD_CF:ydm}\big) are modeled as Gaussian noise \cite{day2012full2}. Specifically, $z^{\text{SI}}_d \sim \mathcal{CN}(0,\mathbb{E}\left[\bar{z}^{\text{SI}}_d\left(\bar{z}^{\text{SI}}_d\right)^\ast\right])$ with $\bar{z}^{\text{SI}}_d=h_{dd}\sum_{\bar{m} \in \bar{\mathcal{M}}}\mu_{d\bar{m}}\sqrt{p_{d\bar{m}}}x_d[\bar{m}]$, where $x_d[\bar{m}]$ denotes the data transmitted on the $\bar{m}$-th UL subcarrier by the $d$-th MS, $p_{d\bar{m}}$ denotes the transmitted power. $z^{\text{IMI}}_d \sim \mathcal{CN}(0,\xi_d^{\text{IMI}}\mathbb{E}\left[\bar{z}^{\text{IMI}}_d\left(\bar{z}^{\text{IMI}}_d\right)^\ast\right])$ with $\bar{z}^{\text{IMI}}_d=\sum_{d^{\prime} \in \mathcal{D}\backslash \left\{d\right\}}\sum_{\bar{m} \in \bar{\mathcal{M}}}\mu_{d^{\prime}\bar{m}}\sqrt{p_{d^{\prime}\bar{m}}}h_{dd^{\prime}}[\bar{m}]x_{d^{\prime}}[\bar{m}]$, where $\xi_d^{\text{IMI}}$ denotes the residual IMI level at MS $d$.

Based on \eqref{eq:MDD_CF:ydm}, it can be shown that the received SINR on the $m$-th DL subcarrier at the $d$-th MS is given by
\begin{equation}\label{eq:MDD_CF:dlSINR}
\text{SINR}_{d,m} = \frac{\left|\sum_{l \in \mathcal{L}}\lambda_{ld}\mu_{ldm} \sqrt{p_{ldm}}\pmb{h}_{ld}^H[m]\pmb{f}_{ld}[m]\right|^2}{\text{MUI}_{d,m}+ \text{var}\left\{z^{\text{SI}}_d\right\} + \text{var}\left\{z^{\text{IMI}}_d\right\}+\sigma^2},
\end{equation} 
where 
\begin{equation}
\text{MUI}_{d,m}=\sum_{l \in \mathcal{L}}\sum_{d^{\prime} \in \mathcal{D}\backslash \left\{d\right\}}\lambda_{ld^{\prime}}\mu_{ld^{\prime}m}p_{ld^{\prime}m}\left|\pmb{h}_{ld}^H[m]\pmb{f}_{ld^{\prime}}[m]\right|^2.
\end{equation}

\subsection{Uplink Transmission}
The received UL signal by the $l$-th AP from the $\bar{m}$ subcarrier of MS $d$ can be expressed as
\begin{align}\label{eq:MDD_CF:ylbarm}
\pmb{y}_l[\bar{m}] &= \underbrace{\sum_{d \in \mathcal{D}}\mu_{d\bar{m}}\sqrt{p_{d\bar{m}}}\pmb{h}_{ld}[\bar{m}]x_d[\bar{m}]}_{\text{Desired signal + MUI}}+ \pmb{z}^{\text{SI}}_l + \pmb{z}^{\text{IAI}}_l + \pmb{n}_l,
\end{align}
where $p_{d\bar{m}}$ denotes the power allocated by MS $d$ to the $\bar{m}$-th UL subcarrier, which satisfies $\sum_{\bar{m} \in \bar{\mathcal{M}}}\mu_{d\bar{m}}p_{d\bar{m}}\leq P_d$. Similar to the received signals at MSs, the residual interference due to the SI and IAI are modeled as Gaussian noise, where $\pmb{z}^{\text{SI}}_l \sim \mathcal{CN}\left(0,\text{diag}\left(\mathbb{E}\left[\bar{\pmb{z}}^{\text{SI}}_l\left(\bar{\pmb{z}}^{\text{SI}}_l\right)^H\right]\right)\right)$ with $\bar{\pmb{z}}^{\text{SI}}_l=\pmb{H}_{ll}\sum_{m \in \mathcal{M}}\pmb{s}_l[m]$, and $\pmb{z}^{\text{IAI}}_l \sim \mathcal{CN}\left(0,\xi_l^{\text{IAI}}\text{diag}\left(\mathbb{E}\left[\bar{\pmb{z}}^{\text{IAI}}_l\left(\bar{\pmb{z}}^{\text{IAI}}_l\right)^H\right]\right)\right)$ with $\bar{\pmb{z}}^{\text{IAI}}_l=\sum_{l^{\prime} \in \mathcal{L}\backslash \left\{l\right\}}\sum_{m \in \mathcal{M}}\pmb{H}_{ll^{\prime}}[m]\pmb{s}_{l^{\prime}}[m]$, where $\xi_l^{\text{IAI}}$ denotes the residual IAI level at the $l$-th AP.
\begin{remark}
$\pmb{z}^{\text{IMI}}_d$ and $\pmb{z}^{\text{IAI}}_l$ in \eqref{eq:MDD_CF:ydm} and \eqref{eq:MDD_CF:ylbarm} are mainly attributed to the large-scale fading of the interfering links, power-allocation and the residual IAI/IMI after the supplementary mitigation of IAI/IMI in propagation/analog- and digital-domain. In distributed CF systems, since each AP works independently, the propagation/analog-domain IAI suppression can hardly rely on the coordinated transmit beamforming and successive interference cancellation \cite{nguyen2020spectral,kusaladharma2021achievable}. Furthermore,  the passive methods, such as antenna cross-polarization, beam separation and absorber can only provide a small part of IAI suppression \cite{kolodziej2019band}. Additionally, due to the employment of single-antenna, the MSs with less powerful baseband processor may fail to handle any IMI in the propagation/analog domain. However, as seen in \eqref{eq:MDD_CF:ydm} and \eqref{eq:MDD_CF:ylbarm}, in our porposed MDD-CF scheme, the IAI/IMI are mutually orthogonal to the desired UL/DL signals in digital domain. Therefore, aided by the large-scale fading and power-allocation to limit the received signal within the effective dynamic range of ADC, the followed FFT operation in digital domain can provide extra mitigation of IAI/IMI~\cite{sabharwal2014band}. 
\end{remark}

We should mention that like any other multicarrier systems, the FFT operation for the mitigation of residual SI and IAI/IMI requires accurate time synchronization \cite{yang2009multicarrier}. As we know, it is relatively easy to achieve time synchronization between transmitter and receiver at one side or between different APs via low-latency fronthauls \cite{jeong2020frequency}, while the synchronization between MSs is challenging and the excessive time synchronization error beyond the allowable time window may degrade the performance of the FFT-relied interference cancellation. However, this issue is beyond the scope of this paper, and will be studied in our future work.    

Due to the feature of distributed operation in our proposed system, each AP firstly processes the received signals from MSs using the local combining vectors, yielding $\tilde{y}_l[\bar{m}]=\pmb{w}_{ld}^H[\bar{m}]\pmb{y}_l[\bar{m}]$, where $\pmb{w}_{ld}[\bar{m}]$ denotes the local combining vector of AP $l$ for detecting MS $d$. Then, the local estimated data by all APs are further collected by the CPU for final processing, which can be expressed as $y_{\text{cpu}}[\bar{m}]=\sum_{l \in \mathcal{L}}\tilde{y}_l[\bar{m}]$. The SINR obtained by the CPU for detecting the data transmitted on the UL subcarrier $\bar{m}$ of MS $d$ can be expressed as
\begin{equation}\label{eq:MDD_CF:ulSINR}
\text{SINR}_{d,\bar{m}}=\frac{\mu_{d\bar{m}}p_{d\bar{m}}\left|\tilde{\pmb{w}}_d[\bar{m}]\tilde{\pmb{h}}_d[\bar{m}]\right|^2}{\text{MUI}_{d,\bar{m}}+ \text{SI}_{d,\bar{m}}+\text{IAI}_{d,\bar{m}}+\sigma^2\left\|\tilde{\pmb{w}}_d[\bar{m}]\right\|^2},
\end{equation}  
where 
\begin{align}
\tilde{\pmb{w}}_d[\bar{m}]&=\left[\pmb{w}_{1d}^H[\bar{m}],...,\pmb{w}_{Ld}^H[\bar{m}]\right] \in \mathbb{C}^{1 \times NL}, \nonumber \\
\tilde{\pmb{h}}_d[\bar{m}]&=[\pmb{h}_{1d}^H[\bar{m}],...,\pmb{h}_{Ld}^H[\bar{m}]]^H \in \mathbb{C}^{NL \times 1}, \nonumber \\
\text{MUI}_{d,\bar{m}}&=\sum_{d^{\prime} \in \mathcal{D}\backslash \left\{d\right\}}\mu_{d^{\prime}\bar{m}}p_{d^{\prime}\bar{m}}\left|\tilde{\pmb{w}}_d[\bar{m}]\tilde{\pmb{h}}_{d^{\prime}}[\bar{m}]\right|^2, \nonumber \\
\text{SI}_{d,\bar{m}}&=\sum_{l \in \mathcal{L}}\mathbb{E}\left[\left\|\pmb{w}_{ld}^H[\bar{m}]\pmb{z}^{\text{SI}}_l\right\|^2\right], \nonumber \\
\text{IAI}_{d,\bar{m}}&=\sum_{l \in \mathcal{L}}\mathbb{E}\left[\left\|\pmb{w}_{ld}^H[\bar{m}]\pmb{z}^{\text{IAI}}_l\right\|^2\right].
\end{align}


\subsection{Beamforming Strategy}
In this paper, the ZF beamforming strategy is chosen for transmitting and receiving at APs. Generally speaking, MMSE beamforming outperforms ZF beamforming when perfect CSI is available, but when considering the multi-MS interfernece suppression, computation complexity as well as concise formulation, ZF is applied in the following analysis, and it can be easily substituted by MMSE in our proposed system. Based on the ZF principle \cite{jiang2011performance}, the precoder/combiner at the $l$-th AP, i.e., $\pmb{F}^{\text{ZF}}_l[m]=\left[\pmb{f}_{l1}^{\text{ZF}}[m],...,\pmb{f}_{lD}^{\text{ZF}}[m]\right]$ and $\pmb{W}^{\text{ZF}}_l[\bar{m}]=\left[\pmb{w}_{l1}^{\text{ZF}}[\bar{m}],...,\pmb{w}_{lD}^{\text{ZF}}[\bar{m}]\right]$, can be derived as $\pmb{F}^{\text{ZF}}_l[m]=\pmb{H}_{l}^H[m]\left(\pmb{H}_{l}[m]\pmb{H}_{l}^H[m]\right)^{-1}$ and $\pmb{W}^{\text{ZF}}_l[\bar{m}]=\pmb{H}_{l}[\bar{m}]\left(\pmb{H}_{l}^H[\bar{m}]\pmb{H}_{l}[\bar{m}]\right)^{-1}$, respectively, where $\pmb{H}_{l}[m]=\left[\pmb{h}_{l1}[m],...,\pmb{h}_{lD}[m]\right]^H$,  $\pmb{H}_{l}[\bar{m}]=\left[\pmb{h}_{l1}[\bar{m}],...,\pmb{h}_{lD}[\bar{m}]\right]$. Note that, in order to ensure that the MUI is fully suppressed, the implementation of ZF beamforming should adhere to the constraint of $N\geq D$.\footnote{For the sake of convenience, we assume that each AP is employed with sufficient antennas so as to suppress the interference that itself generates. Although an AP is expected to be equipped with a small number of antennas in CF systems, our assumption is still practical, as each AP can be treated as a secondary central unit controlling $N$ single-antenna APs operated in a centralized mode through fronthaul connections.} In this case, the MUI terms in \eqref{eq:MDD_CF:dlSINR} and \eqref{eq:MDD_CF:ulSINR} are equal to zero. Therefore, the $\text{SINR}_{d,m}$ and $\text{SINR}_{d,\bar{m}}$ can be rewritten as follows
\begin{align}\label{eq:MDDCF:SINRSim}
\text{SINR}_{d,m}&=\frac{\left|\sum_{l \in \mathcal{L}}\lambda_{ld}\mu_{ldm} \sqrt{p_{ldm}}\omega_{ldm}\right|^2}{\xi_d^{\text{SI}}\Theta_{\text{DL}}+\sigma^2}, \nonumber\\
\text{SINR}_{d,\bar{m}}&=\frac{\mu_{d\bar{m}}p_{d\bar{m}}L^2}{\sum_{l \in \mathcal{L}}\upsilon_{ld\bar{m}}\left(\xi_l^{\text{SI}}\Theta_{\text{UL}}+\sigma^2\right)}, 
\end{align} 
where
\begin{align}
&\omega_{ldm}=\frac{1}{\left\|\pmb{f}_{ld}^{\text{ZF}}[m]\right\|_2}, \upsilon_{ld\bar{m}}=\left\|\pmb{w}_{ld}^{\text{ZF}}[\bar{m}]\right\|_2^2,\nonumber \\
&\Theta_{\text{DL}}=\sum_{\bar{m} \in \bar{\mathcal{M}}}\mu_{d\bar{m}}p_{d\bar{m}} +\frac{\xi_d^{\text{IMI}}}{\xi_d^{\text{SI}}} \sum_{d^{\prime} \in \mathcal{D}\backslash \left\{d\right\}}\sum_{\bar{m} \in \bar{\mathcal{M}}}\frac{\beta_{dd^{\prime}}}{M_{\text{sum}}}\mu_{d^{\prime}\bar{m}}p_{d^{\prime}\bar{m}}, \nonumber \\
&\Theta_{\text{UL}}=\sum_{m \in \mathcal{M}}\sum_{d \in \mathcal{D}} \lambda_{ld}\mu_{ldm} p_{ldm} \nonumber \\
&+\frac{\xi_l^{\text{IAI}}}{\xi_l^{\text{SI}}}\sum_{l^{\prime} \in \mathcal{L}\backslash \left\{l\right\}}\sum_{m \in \mathcal{M}}\sum_{d \in \mathcal{D}} \frac{\beta_{ll^{\prime}}}{M_{\text{sum}}}\lambda_{l^{\prime}d}\mu_{l^{\prime}dm} p_{l^{\prime}dm}.
\end{align}
For the details of simplification, please refer to Appendix \ref{Appen:MDD_CF_SINRSim}.

Consequently, the average SE of the MDD-CF scheme in nats/s/Hz can be expressed as
\begin{align}
\label{eq:MDDCF:MDDSE}
\Lambda_{\text{SE}}&=\left(1-\frac{\gamma^{\text{P}}}{T_{\text{c}}}\right) \frac{1}{M_{\text{sum}}}\sum_{d \in \mathcal{D}} \Big(\sum_{m \in \mathcal{M}}R(\text{SINR}_{d,m}), \nonumber \\
&+\sum_{\bar{m} \in \mathcal{\bar{M}}}R(\text{SINR}_{d,\bar{m}})\Big)
\end{align}
where $R(x)\triangleq ln(1+x)$.

\section{SE Optimization within single CT Interval}\label{sec:MDD_CF:seo}
In this section, we aim to maximize the SE over one CT interval in the MDD-CF scheme, as shown in Fig. \ref{figure-MDDCF-Coh}. Given the ZF beamforming, the optimization problem can be stated as:
\begin{subequations}
\label{eq:MDDCF:SE_formulation}
\begin{align}
&\max_{\lambda_{ld}, \mu_{ldm}, \mu_{d\bar{m}}, p_{ldm}, p_{d\bar{m}}} \Lambda_{\text{SE}}  \\ 
\text{s.t.} \ \ &~~\lambda_{ld} \in \left\{0,1\right\}, \ \forall l \in \mathcal{L}, d \in \mathcal{D}, \\
&~~\mu_{ldm} \in \left\{0,1\right\},\ \forall l \in \mathcal{L}, d \in \mathcal{D}, m \in \mathcal{M}, \\
&~~\mu_{d\bar{m}} \in \left\{0,1\right\},\ \forall d \in \mathcal{D}, \bar{m} \in \mathcal{\bar{M}},\\
&~~\sum_{m \in \mathcal{M}}\sum_{d \in \mathcal{D}} \lambda_{ld}\mu_{ldm} p_{ldm}\leq P_l, \ \forall l \in \mathcal{L}, \\
&~~\sum_{\bar{m} \in \bar{\mathcal{M}}}\mu_{d\bar{m}}p_{d\bar{m}}\leq P_d, \ \forall d \in \mathcal{D}, \\
&~~\sum_{m \in \mathcal{M}}R(\text{SINR}_{d,m}) \geq \chi_{\text{DL}}, \ \forall d  \in \mathcal{D}, \\
&~~\sum_{\bar{m} \in \mathcal{\bar{M}}}R(\text{SINR}_{d,\bar{m}}) \geq \chi_{\text{UL}}, \ \forall d  \in \mathcal{D},
\end{align}
\end{subequations}
where the constraints of (\ref{eq:MDDCF:SE_formulation}g) and (\ref{eq:MDDCF:SE_formulation}h) are applied to guarantee the MS's QoS requirements for DL and UL so as to avoid unbalanced greedy resource-allocation among MSs. It can be observed that (\ref{eq:MDDCF:SE_formulation}) is an optimization problem of AP-selection and resource-allocation, which is hard to solve because the binary variables \big(i.e., $\pmb{\lambda}=\left\{\lambda_{ld}\right\}, \pmb{\mu}=\left(\left\{\mu_{ldm}\right\}, \left\{\mu_{d\bar{m}}\right\}\right)$\big) are tightly coupled with the continuous variables $\pmb{p}=(\left\{p_{ldm}\right\}, \left\{p_{d\bar{m}}\right\})$. To circumvent this problem, in what follows, we first focus on the reduction and approximation of the involved binary variables.  

\subsection{Reduction of Binary Variables}
\subsubsection{Reduction of $\mu_{ldm}$ and $\mu_{d\bar{m}}$}
Let us first consider the case of $\mu_{ldm}$. According to \eqref{eq:MDD_CF:dlSINR}, the relationship between $\mu_{ldm}$ and $p_{ldm}$ at the $l$-th AP can be provided by the following lemma.
\begin{lemma}\label{lemma:MDD_CF:mup}
For the potentially optimal solution of \eqref{eq:MDDCF:SE_formulation}, the only feasible combinations of $\mu_{ldm}$ and $p_{ldm}$ are $(\mu_{ldm}^{\ast},p_{ldm}^{\ast})\in\left\{(0,0),(1,\tilde{p}_{ldm})\right\}$, where $\tilde{p}_{ldm}\neq 0$. 
\end{lemma}
\begin{proof}
For the potentially optimal solution of \eqref{eq:MDDCF:SE_formulation}, the possible combinations of $\mu_{ldm}$ and $p_{ldm}$ are $(\mu_{ldm}^{\ast},p_{ldm}^{\ast})\in\left\{(0,0),(0,\tilde{p}_{ldm}),(1,0),(1,\tilde{p}_{ldm})\right\}$, where $\tilde{p}_{ldm}\neq 0$. Then, it can be easily found that $\Lambda_{\text{SE}}\big(\left(\lambda,\pmb{\mu},\pmb{p}\right)|\mu_{ldm}=0 \ \& \ p_{ldm}=0 \big)=\Lambda_{\text{SE}}\big(\left(\lambda,\pmb{\mu},\pmb{p}\right)| \ \mu_{ldm}p_{ldm}=0 \ \& \ \mu_{ldm}+p_{ldm}\neq 0 \big)$, since $\mu_{ldm}$ and $p_{ldm}$ are tightly coupled in \eqref{eq:MDDCF:SINRSim}. When $\mu_{ldm}$ or $p_{ldm}$ is equal to 0, $\text{SINR}_{d,m}$ and $\text{SINR}_{d,\bar{m}}$ remain unchanged as $\mu_{ldm}p_{ldm}=0$.  In conclusion, due to the special relationship between $\mu_{ldm}$ and $p_{ldm}$, these two variables in the optimal solution of \eqref{eq:MDDCF:SE_formulation} can only be either $(0,0)$ or $(1,\tilde{p}_{ldm})$. This completes the proof.
\end{proof}
Based on Lemma \ref{lemma:MDD_CF:mup}, after the optimization problem \eqref{eq:MDDCF:SE_formulation} is solved, all $\mu_{ldm}$ can be subsequently obtained from the optimal power-allocation, i.e., $p_{ldm}^{\ast}$, which can be given as
\begin{equation}\label{eq:MDD_CF:ldmpower}
\mu_{ldm}=
\begin{cases}
0, &\frac{p_{ldm}^{\ast}}{P_l}<\kappa \\
1, &\frac{p_{ldm}^{\ast}}{P_l}\geq \kappa
\end{cases},
\end{equation}
where $\kappa$ is a very small number, implying that a small value of $p_{ldm}^{\ast}$ can be deemed as zero. Analogously, we can apply Lemma \ref{lemma:MDD_CF:mup} and \eqref{eq:MDD_CF:ldmpower} to derive $p_{d\bar{m}}$ and $\mu_{d\bar{m}}$.

\subsubsection{Reduction of $\lambda_{ld}$}
The binary variable $\lambda_{ld}$ denotes the association status between AP $l$ and MS $d$. Intuitively, once any of the DL subcarriers is activated at AP $l$ for transmitting data to MS $d$, i.e., $\exists m\in\mathcal{M}, \mu_{ldm}=1$, the communication link between AP $l$ and MS $d$ is established. Therefore, $\lambda_{ld}$ can be obtained as
\begin{equation}\label{eq:MDDCF:lamdald}
\lambda_{ld}=\text{max}\left\{\mu_{ldm}|m\in \mathcal{M}\right\}, \ \forall l \in \mathcal{L}, \ d \in \mathcal{D}.  
\end{equation}


\subsection{Maximization of SE Based on Quadratic Transform}
After the reduction of the binary variables, as shown in Section \ref{sec:MDD_CF:seo}.A, the SE optimization problem can be transformed to a relatively simple form as
\begin{subequations}
\label{eq:MDDCF:SE_simpli}
\begin{align}
&\max_{\pmb{p}} \ \Lambda_{\text{SE}}  \\ 
\text{s.t.} \ &~~(\ref{eq:MDDCF:SE_formulation}\text{b}), (\ref{eq:MDDCF:SE_formulation}\text{c}), (\ref{eq:MDDCF:SE_formulation}\text{d}), (\ref{eq:MDDCF:SE_formulation}\text{e}), (\ref{eq:MDDCF:SE_formulation}\text{f}),\\
&~~\text{SINR}_{d,m}\geq e^{\frac{\chi_{\text{DL}}}{M_d}}-1, \ \forall d \in \mathcal{D}, m \in \mathcal{M},\\
&~~\text{SINR}_{d,\bar{m}}\geq e^{\frac{\chi_{\text{UL}}}{\bar{M}_d}}-1, \ \forall d \in \mathcal{D}, \bar{m} \in \mathcal{\bar{M}},
\end{align}
\end{subequations}
where $M_d$ and $\bar{M}_d$ denote the numbers of DL and UL subcarriers assigned to MS $d$, respectively. During the optimization, the values of $\pmb{\lambda}$ and $\pmb{\mu}$ are initialized to $\pmb{1}$ and then iteratively updated with the results of power-allocation. Furthermore, the ZF precoder/combiner are also re-computed during each iteration according to the updated results of AP-MS connection and subcarrier-allocation. Note that, the original constraints (\ref{eq:MDDCF:SE_formulation}g) and (\ref{eq:MDDCF:SE_formulation}h) contain a sum of $M$ and $\bar{M}$ nonconvex components for each MS $d$, respectively, resulting in extremely high complexity. Hence, to make the optimization tractable, these two constraints are simplified to (\ref{eq:MDDCF:SE_simpli}c) and (\ref{eq:MDDCF:SE_simpli}d). 

It can be shown that the scaled-down objective function (\ref{eq:MDDCF:SE_simpli}a) belongs to the general multiple-ratio concave-convex fractional programming (CCFP) problem \cite{shen2018fractional}. However, the constraints of (\ref{eq:MDDCF:SE_simpli}c) and (\ref{eq:MDDCF:SE_simpli}d) are still nonconvex, which have to be approximated by the convex ones.

Specifically, based on \eqref{eq:MDDCF:SINRSim}, the constraint (\ref{eq:MDDCF:SE_simpli}c) can be equivalently written as
\begin{subnumcases}{\label{eq:MDDCF:SINRdmEqu} (\ref{eq:MDDCF:SE_simpli}\text{c}) \Longleftrightarrow}
&$\text{SINR}_{d,m} \triangleq \varpi^2_{d,m} / \psi_{d,m} \geq e^{\frac{\chi_{\text{DL}}}{M_d}}-1, $ \label{eq:MDDCF:SINRdmEqu:a} \\
&$0<\varpi_{d,m} \leq \sum_{l \in \mathcal{L}}\sqrt{p_{ldm}}\omega_{ldm}, $ \label{eq:MDDCF:SINRdmEqu:b} \\
&$\psi_{d,m} \geq \xi_d^{\text{SI}}\Theta_{\text{DL}}+\sigma^2, $ \label{eq:MDDCF:SINRdmEqu:c}
\end{subnumcases}
where $\varpi_{d,m}$ and $\psi_{d,m}$ are new variables, while \eqref{eq:MDDCF:SINRdmEqu:b} and \eqref{eq:MDDCF:SINRdmEqu:c} are linear constraints. For \eqref{eq:MDDCF:SINRdmEqu:a}, since the function $f_{\text{ca}}(\varpi_{d,m},\psi_{d,m})\triangleq \varpi^2_{d,m} / \psi_{d,m}$ with $(\varpi_{d,m},\psi_{d,m})\in \mathbb{R}_{++}^2$ is convex, it can be approximated using the successive convex approximation (SCA) properties as \cite{marks1978general}
\begin{align}
&f_{\text{ca}}(\varpi_{d,m},\psi_{d,m})\nonumber \\
&\geq \frac{2\varpi_{d,m}^{(t)}}{\psi_{d,m}^{(t)}}\varpi_{d,m}-\frac{(\varpi_{d,m}^{(t)})^2}{(\psi_{d,m}^{(t)})^2}\psi_{d,m} :=f_{\text{ca}}^{(t)}(\varpi_{d,m},\psi_{d,m}),
\end{align}
where $(\varpi_{d,m}^{(t)},\psi_{d,m}^{(t)})$ is the feasible point obtained at the $t$-th iteration. Therefore, \eqref{eq:MDDCF:SINRdmEqu:a} can be substituted by the new constraint given by
\begin{equation}\label{eq:MDDCF:SINRCA}
f_{\text{ca}}^{(t)}(\varpi_{d,m},\psi_{d,m})\geq e^{\frac{\chi_{\text{DL}}}{M_d}}-1, \ \forall d \in \mathcal{D}, m \in \mathcal{M},
\end{equation}

Following the same spirit, (\ref{eq:MDDCF:SE_simpli}d) can be replaced by the following convex constraints:
\begin{subnumcases}{\label{eq:MDDCF:SINRdbarmEqu} (\ref{eq:MDDCF:SE_simpli}\text{d}) \Longleftrightarrow}
&$f_{\text{ca}}^{(t)}(\sqrt{\varpi_{d,\bar{m}}},\psi_{d,\bar{m}})\geq  e^{\frac{\chi_{\text{UL}}}{\bar{M}_d}}-1,$ \label{eq:MDDCF:SINRdbarmEqu:a} \\
&$0<\varpi_{d,\bar{m}} \leq p_{d\bar{m}}L^2,$ \label{eq:MDDCF:SINRdbarmEqu:b} \\
&$\psi_{d,m} \geq \sum_{l \in \mathcal{L}}\upsilon_{ld\bar{m}}\left(\xi_l^{\text{SI}}\Theta_{\text{UL}}+\sigma^2\right).$ \label{eq:MDDCF:SINRdbarmEqu:c}
\end{subnumcases}

To this point, all the constraints in \eqref{eq:MDDCF:SE_simpli} are convex, and we can now apply the QT to deal with \eqref{eq:MDDCF:SE_simpli}. According to \cite[Corollary 2]{shen2018fractional}, the sum-of-functions-of-ratio problem in (\ref{eq:MDDCF:SE_simpli}a), i.e.,
\begin{subequations}
\begin{align}
&\max_{\pmb{p}} \nonumber \\
&\frac{1}{M_{\text{sum}}}\sum_{d=1}^D \big(\sum_{m=1}^M R_{d,m}(\frac{A_{d,m}(\pmb{p})}{B_{d,m}(\pmb{p})})+\sum_{\bar{m}=1}^{\bar{M}} R_{d,\bar{m}}(\frac{A_{d,\bar{m}}(\pmb{p})}{B_{d,\bar{m}}(\pmb{p})})\big) \\
&\text{s.t.} \ (\ref{eq:MDDCF:SE_simpli}\text{b}), \eqref{eq:MDDCF:SINRdmEqu:b}, \eqref{eq:MDDCF:SINRdmEqu:c}, \eqref{eq:MDDCF:SINRCA}, \eqref{eq:MDDCF:SINRdbarmEqu},
\end{align}
\end{subequations}
can be equivalently described as
\begin{subequations}
\label{eq:MDDCF:QT}
\begin{align}
&\max_{\pmb{p}} \frac{1}{M_{\text{sum}}}\sum_{d=1}^D \bigg(\sum_{m=1}^M R_{d,m}\big(2z_{dm}\sqrt{A_{d,m}(\pmb{p})}-z_{dm}^2B_{d,m}(\pmb{p})\big) \nonumber \\
&+\sum_{\bar{m}=1}^{\bar{M}} R_{d,\bar{m}}(2z_{d\bar{m}}\sqrt{A_{d,\bar{m}}(\pmb{p})}-z_{d\bar{m}}^2 B_{d,\bar{m}}(\pmb{p})\bigg) \\
&\text{s.t.} \ z_{dm} \in \mathbb{R}, \ \forall d \in \mathcal{D}, m \in \mathcal{M} \ , \ z_{d\bar{m}} \in \mathbb{R}, \ \forall d \in \mathcal{D}, \bar{m} \in \mathcal{\bar{M}},\\
&(\ref{eq:MDDCF:SE_simpli}\text{b}), \eqref{eq:MDDCF:SINRdmEqu:b}, \eqref{eq:MDDCF:SINRdmEqu:c}, \eqref{eq:MDDCF:SINRCA}, \eqref{eq:MDDCF:SINRdbarmEqu},
\end{align}
\end{subequations}
where $A(\pmb{p})$ and $B(\pmb{p})$ denote the numerator and denominator of \eqref{eq:MDDCF:SINRSim}, respectively. Since $R(x)$ for all $d$, $m$ and $\bar{m}$ is non-decreasing and concave, and for a given $\pmb{p}$, $\frac{A(\pmb{p})}{B(\pmb{p})}$ is in the concave-convex form, the optimal $\pmb{z}=\left(\left\{z_{dm}\right\},\left\{z_{d\bar{m}}\right\}\right)$ can be obtained as $\pmb{z}^{\ast}=\frac{\sqrt{A(\pmb{p})}}{B(\pmb{p})}$. Then, for a given $\pmb{z}$, the problem \eqref{eq:MDDCF:QT} is a concave maximization problem over $\pmb{p}$. Therefore, the overall problem can essentially be solved by a block coordinate ascent algorithm, with $\pmb{z}$ and $\pmb{p}$ iteratively optimized, until the optimization converges to a local optimum. Note that our QT-assisted algorithm mainly relies on the QT and SCA methods, and the details for the proof of their convergence can be found in \cite{shen2018fractional} and \cite{marks1978general}. Furthermore, the convergence properties and the complexity of the algorithm will be studied in Section \ref{sec:MDD_CF:sim}. In summary, the overall optimization algorithm for SE maximization is stated as Algorithm \ref{MDDCF:al1}.

\begin{algorithm}
\caption{QT-assisted Algorithm for SE maximization in MDD-CF scheme} 
\label{MDDCF:al1}
\textbf{Initialization:} \\
Set $\pmb{\lambda}=\pmb{1}, \pmb{\mu}=\pmb{1}$\;
Compute $\left\{\omega_{ldm}\right\}, \left\{\upsilon_{ld\bar{m}}\right\}, \forall l, d, m, \bar{m}$\;
Set $t=0$, $\pmb{p}^{(0)}$ to a feasible value, $\varpi_{d,m}^{(0)}=1, \varpi_{d,\bar{m}}^{(0)}=1, \psi_{d,m}^{(0)}=1, \psi_{d,\bar{m}}^{(0)}=1, \forall d, m, \bar{m}$\;
\QT{}{
\Repeat{$\textup{Convergence}$}{
Compute $\pmb{z}^{(t)}$ using $\pmb{z}^{(t)}=\frac{\sqrt{A(\pmb{p}^{(t)})}}{B(\pmb{p}^{(t)})}$, for a fixed $\pmb{p}^{(t)}$ \;
Update $\pmb{p}^{(t+1)}$ by solving \eqref{eq:MDDCF:QT}, for a fixed $\pmb{z}^{(t)}$  \;
Update $\varpi_{d,m}^{(t+1)}, \varpi_{d,\bar{m}}^{(t+1)}, \psi_{d,m}^{(t+1)}, \psi_{d,\bar{m}}^{(t+1)}$\;
Set $t = t + 1$;
}
}
Update $\pmb{\mu}$ and $\pmb{\lambda}$ using \eqref{eq:MDD_CF:ldmpower} and \eqref{eq:MDDCF:lamdald}, respectively\;
Repeat Step 3 to Step 13 until $\pmb{\lambda}$ and $\pmb{\mu}$ are stable, and obtain the optimal SE, i.e., $\Lambda_{\text{SE}}^{\ast}$.

\KwOut{$\pmb{\lambda}, \pmb{\mu}, \pmb{p}, \Lambda_{\text{SE}}^{\ast}$}
\end{algorithm}

\section{The SE Optimization within Radio Frame}\label{sec:MDD_CF:serf}
The optimization problem addressed in the previous section only considered the single CT interval and assumed perfect CSI. In this section, we extend our studies by considering a more complicated scenario, where the radio frame with imperfect channel estimation is assumed. More specifically, as shown in Fig. \ref{figure-MDDCF-MulCoh}, in our proposed frame structure, starting from the second CT interval of $T_\text{c}^2$, IBFD- and MDD-CF schemes can exploit an extra $\gamma^{\text{P}}$ of time for DL transmission at the beginning of the interval during UL training, owing to their FD feature. Therefore, for the convenience of illustration, we assume that after $T_\text{c}^1$, each TPCT interval consists of two transmission phases, namely Phase I with $\gamma^{\text{P}}$ of supplementary DL transmission and UL training, and Phase II with $(T_{\text{c}}-\gamma^{\text{P}})$ of DL/UL simultaneous transmissions. By contrast, in the counterpart TDD-CF scheme, since pilots and data are sequentially transmitted, the SE performance may deteriorate quickly in some cases. For example, when the relative velocity between MS and AP increases, UL training has to be implemented more frequently as the result of the shorter CT intervals, which unavoidably leads to the SE degradation.

However, in order to unleash the full advantages of FD, two paramount challenges need to be addressed in this case. Firstly, since the channel varies continuously from one CT interval to the next, the performance of the supplementary DL transmission in Phase I hinges on the predicted channel. For this problem, in literature, there are various channel prediction methods, e.g., Wiener filter\cite{truong2013effects}, Kalman filter\cite{kashyap2017performance} and deep learning \cite{yuan2020machine}, which can be applied to predict the required channels with high accuracy. Secondly, as depicted in Fig. \ref{figure-MDDCF-MulCoh}, although the extra DL transmission during Phase I can increase SE due to the added transmission time, it may cause interference on the receiving of UL pilots, and hence affect the channel estimation accuracy, which in turn leads to degraded performance in Phase II. Therefore, there exists a trade-off between the accuracy of channel acquisition and the SE provided by the supplementary DL transmission. To this end, we shall focus on this trade-off problem in the sequel under the assumption that the CSI used for the supplementary DL transmission is predicted using the Wiener filter.

\begin{figure}
\centering
\includegraphics[width=0.99\linewidth]{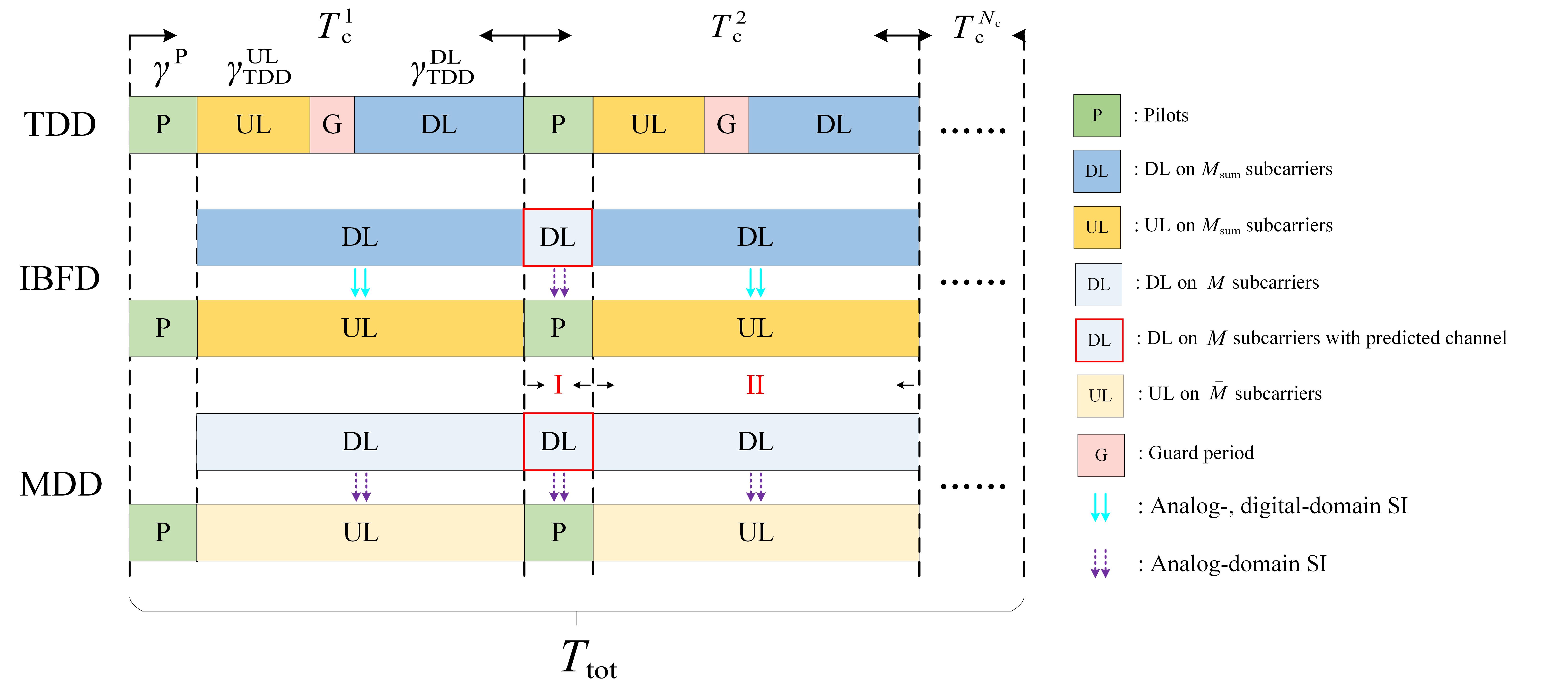}
\caption{Structure of radio frame of TDD-, IBFD- and MDD- schemes, where $N_c$ denotes the number of CT intervals within one radio frame.}
\label{figure-MDDCF-MulCoh}
\end{figure}

\subsection{Channel Estimation in MDD-CF Scheme}
Let us assume that all MSs synchronously transmit their frequency-domain pilot sequences (FDPS) over the $\bar{M}$ UL subcarriers, while all APs transmit DL data to MSs over the DL subcarriers. Let the FDPS transmitted by the $d$-th MS be expressed as $\pmb{x}_{d}^{\text{p}}=\left[x_{d}^{\text{p}}[1],...,x_{d}^{\text{p}}[\bar{m}],...,x_{d}^{\text{p}}[\bar{M}]\right]^T$. Then, the received training signal at the $n$-th antenna of AP $l$ can be written as
\begin{equation}\label{eq:MDD_CF:ypnCE}
\pmb{y}_{l}^{n} = \sum\limits_{d=1}^D\sqrt{p_{d}^{\text{p}}}\pmb{X}_{d}^{\text{p}}{\pmb{\Phi}_{\text{UL}}}\pmb{F}\pmb{\varPsi}\pmb{g}_{ld}^n+\pmb{z}^{n,\text{SI}}_l+\pmb{z}^{n,\text{IAI}}_l+\pmb{n}^n_l,
\end{equation}
where $p_{d}^{\text{p}}$ denotes the power assigned for sending pilots on each UL subcarrier by MS $d$, $\pmb{\Phi}_{\text{UL}}=\pmb{I}_{M_{\text{sum}}}^{(\bar{\mathcal{M}},:)}$, $\pmb{X}_{d}^{\text{p}}=\text{diag}\left\{\pmb{x}_{d}^{\text{p}}\right\}$, $\pmb{n}^n_l\sim \mathcal{CN}(\pmb{0}, \sigma^2\pmb{I}_{\bar{M}})$. Note that the transmit power of pilots is assumed to be larger than that of UL data, which is fixed during communication. According to Appendix \ref{Appen:MDD_CF_SINRSim}, we can obtain $\pmb{z}^{n,\text{SI}}_l\sim \mathcal{CN}\left(\pmb{0},\left(\text{cov}\left\{\pmb{z}_l^{\text{SI}}\right\}\right)_{n,n}\pmb{I}_{\bar{M}}\right)$ and $\pmb{z}^{n,\text{IAI}}_l\sim \mathcal{CN}\left(\pmb{0},\left(\text{cov}\left\{\pmb{z}_l^{\text{IAI}}\right\}\right)_{n,n}\pmb{I}_{\bar{M}}\right)$. 

According to our previous study \cite{li2020self}, the FDPS of MS $d$ can be designed as
\begin{equation}
\pmb{x}_{d}^{\text{p}}=\left[1,e^{j2\pi\frac{(d-1)\vartheta}{\bar{M}}},e^{j2\pi\frac{2(d-1)\vartheta}{\bar{M}}},\cdots, e^{j2\pi\frac{(\bar{M}-1)(d-1)\vartheta}{\bar{M}}}\right]^T,
\end{equation}
where $\vartheta=\left\lfloor \frac{\bar{M}}{D}\right\rfloor$. With the assumption that $\bar{M}\geq DU$ {\footnote{This assumption can be easily met in 5G and beyond systems \cite{3gpp2017nr}, as a large number of subcarriers is expected, especially, in the mm-wave and THz bands. Note furthermore that, even if the distribution of MSs become denser or if the number of multipaths increases due to more complicated environments, leading to $\bar{M} < DU$, the LMMSE method proposed in \cite{li2020self} can still guarantee an acceptable estimation performance. However, for the simplicity of analysis, the case of $\bar{M} < DU$ is not further considered in this paper.}} and the $\bar{M}$ UL subcarriers are evenly distributed, $\pmb{J}_d=\pmb{X}_{d}^{\text{p}}{\pmb{\Phi}_{\text{UL}}}\pmb{F}\pmb{\varPsi}$ with respect to $d=1,...,D$ are mutually orthogonal, i.e., 
\begin{align}
\begin{cases}
\pmb{J}_d^H \pmb{J}_d =\frac{\bar{M}}{M_{\text{sum}}} \pmb{I}_U,  \\
\pmb{J}_d^H \pmb{J}_k= \pmb{0}_U,~\forall~ {d}\neq{k}.
\end{cases}
\end{align}
Then, the noisy observation of $\pmb{g}_{n,d}[i]$ can be formed as
\begin{equation}\label{eq:MDD_CF:gndobser}
\tilde{\pmb{y}}_{ld}^n=\pmb{J}_d^H \pmb{y}_{l}^n=\frac{\sqrt{p_{d}^{\text{p}}}\bar{M}}{M_{\text{sum}}}\pmb{g}_{ld}^n+\pmb{J}_d^H (\pmb{z}^{n,\text{SI}}_l+\pmb{z}^{n,\text{IAI}}_l+\pmb{n}^n_l).
\end{equation} 
Correspondingly, the MMSE estimate to $\pmb{g}_{ld}^n$ is given by
\begin{equation}\label{eq:MDD_CF:MMSE}
\hat{\pmb{g}}_{ld}^n=\frac{\frac{\beta_{ld}\sqrt{p_{d}^{\text{p}}}}{U}}{\frac{\beta_{ld}p_{d}^{\text{p}}\bar{M}}{UM_{\text{sum}}}+\left(\xi_l^{\text{SI}}I_l+\sigma^2\right)}\tilde{\pmb{y}}_{ld}^n,
\end{equation} 
where $I_l=\sum_{m \in \mathcal{M}}\sum_{d \in \mathcal{D}} \lambda_{ld}^{\text{I}}\mu_{ldm}^{\text{I}} p_{ldm}^{\text{I}}+\frac{\xi_l^{\text{IAI}}}{\xi_l^{\text{SI}}}\sum_{l^{\prime} \in \mathcal{L}\backslash \left\{l\right\}}\sum_{m \in \mathcal{M}}\sum_{d \in \mathcal{D}} \frac{\beta_{ll^{\prime}}}{M_{\text{sum}}}\lambda_{l^{\prime}d}^{\text{I}}\mu_{l^{\prime}dm}^{\text{I}} p_{l^{\prime}dm}^{\text{I}}$ denotes the summation of SI and IAI. Here the superscript `I' means that the interference arises from the supplementary DL transmission of Phase I, as shown in Figure \ref{figure-MDDCF-MulCoh}. Then, according to the properties of MMSE, $\pmb{g}_{ld}^n$ can be orthogonally decomposed into $\pmb{g}_{ld}^n=\hat{\pmb{g}}_{ld}^n+\pmb{e}_{ld}^n$, where $\pmb{e}_{ld}^n$ is the channel estimation error vector uncorrelated with $\hat{\pmb{g}}_{ld}^n$ , which has the covariance matrix 
\begin{equation}\label{eq:MDD_CF:ErrorMMSE}
\Xi_e = \text{cov}\left\{\pmb{e}_{ld}^n\right\}=\left(\frac{\beta_{ld}}{U}-\frac{\frac{\beta_{ld}^2p_{d}^{\text{p}}\bar{M}}{U^2M_{\text{sum}}}}{\frac{\beta_{ld}p_{d}^{\text{p}}\bar{M}}{UM_{\text{sum}}}+\left(\xi_l^{\text{SI}}I_l+\sigma^2\right)}\right)\pmb{I}_{U}.
\end{equation}

\subsection{SE Maximization within Two-Phase CT Interval }
As the case of the SE optimization in $T_\text{c}^1$ has already been addressed in Section \ref{sec:MDD_CF:serf}, here we only study the case of the new designed TPCT interval, i.e., $T_\text{c}^n, \ n=2,...,N_c$, as shown in Fig. \ref{figure-MDDCF-MulCoh}. 

To begin with, the average SE in Phase I can be expressed as 
\begin{equation}
\Lambda_{\text{SE}}^{\text{I}} = \frac{\gamma^{\text{P}}}{T_{\text{c}}}\frac{1}{M_{\text{sum}}}\sum_{d \in \mathcal{D}} \sum_{m \in \mathcal{M}}R(\check{\text{SINR}}_{d,m}),
\end{equation}
where $\check{\text{SINR}}_{d,m}$ is different from the previous SINR, which depends on the predicted CSI by Wiener filter. Then, the maximization of the average SE in a TPCT interval amounts to the following optimization problem:
\begin{subequations}
\label{eq:MDDCF:SE_sup}
\begin{align}
&\max_{\pmb{p}^{\text{I}},\pmb{p}^{\text{II}}, \pmb{\lambda}^{\text{I}}, \pmb{\mu}^{\text{I}}, \pmb{\lambda}^{\text{II}}, \pmb{\mu}^{\text{II}}}  \ \Lambda_{\text{SE}}^{\text{TPCT}}=\Lambda_{\text{SE}}^{\text{I}}+\Lambda_{\text{SE}}^{\text{II}}  \\ 
\text{s.t.} \ &~~(\ref{eq:MDDCF:SE_formulation}\text{b}), (\ref{eq:MDDCF:SE_formulation}\text{c}), (\ref{eq:MDDCF:SE_formulation}\text{d}), (\ref{eq:MDDCF:SE_formulation}\text{e}), (\ref{eq:MDDCF:SE_formulation}\text{f}), (\ref{eq:MDDCF:SE_simpli}\text{b}), (\ref{eq:MDDCF:SE_simpli}\text{c}), \\
&~~\lambda_{ld}^{\text{I}} \in \left\{0,1\right\}, \ \forall l \in \mathcal{L}, d \in \mathcal{D}, \\
&~~\mu_{ldm}^{\text{I}} \in \left\{0,1\right\},\ \forall l \in \mathcal{L}, d \in \mathcal{D}, m \in \mathcal{M}, \\
&~~\sum_{m \in \mathcal{M}}\sum_{d \in \mathcal{D}} \lambda_{ld}^{\text{I}}\mu_{ldm}^{\text{I}} p_{ldm}^{\text{I}}\leq P_l, \ \forall l \in \mathcal{L},
\end{align}
\end{subequations}
where $\pmb{p}^{\text{I}/\text{II}}=\left\{p_{ldm}^{\text{I}/\text{II}}\right\}$, $\pmb{\lambda}^{\text{I}/\text{II}}=\left\{\lambda_{ld}^{\text{I}/\text{II}}\right\}$, $\pmb{\mu}^{\text{I}/\text{II}}=\left\{\mu_{ldm}^{\text{I}/\text{II}}\right\}$. The $\Lambda_{\text{SE}}^{\text{I}}$ is only related to $\pmb{p}^{\text{I}}$, while the $\Lambda_{\text{SE}}^{\text{II}}$ depends on both $\pmb{p}^{\text{I}}$ and $\pmb{p}^{\text{II}}$. (\ref{eq:MDDCF:SE_formulation}b)-(\ref{eq:MDDCF:SE_formulation}f) and (\ref{eq:MDDCF:SE_sup}c)-(\ref{eq:MDDCF:SE_sup}e) explain the power constraints and binary selections for Phase II and Phase I, respectively. (\ref{eq:MDDCF:SE_simpli}b) and (\ref{eq:MDDCF:SE_simpli}c) are the QoS constraints on Phase II, meaning that the objective in Phase I is to increase the overall SE as much as possible, while the QoS constraints are only imposed on Phase II. However, as shown in \eqref{eq:MDD_CF:MMSE} and \eqref{eq:MDD_CF:ErrorMMSE}, $\pmb{p}^{\text{I}}$ in Phase I is included in the error vector of the time-domain estimated CSI. Hence, it implicitly affects the following DL/UL transmission in Phase II, which makes the above optimization problem intractable to obtain $\pmb{p}^{\text{I}}$ and $\pmb{p}^{\text{II}}$ at the same time. 

In order to circumvent this problem, we introduce a two-step iterative algorithm to split the original optimization problem of \eqref{eq:MDDCF:SE_sup} into two sub-problems, namely the optimizations of $\Lambda_{\text{SE}}^{\text{I}}$ and $\Lambda_{\text{SE}}^{\text{II}}$, which are expressed as follows: 
\begin{subequations}
\label{eq:MDDCF:SE_sub1}
\begin{align}
\text{Sub 1:} \ \ \ &\max_{\pmb{p}^{\text{I}},\pmb{\lambda}^{\text{I}},\pmb{\mu}^{\text{I}}} \ \Lambda_{\text{SE}}^{\text{I}} \\ 
\text{s.t.} \ &~~(\ref{eq:MDDCF:SE_sup}\text{c}), (\ref{eq:MDDCF:SE_sup}\text{d}), (\ref{eq:MDDCF:SE_sup}\text{e}),\\
&~~\xi_l^{\text{SI}}I_l\leq \eta_l \sigma^2, \ \forall l,
\end{align}
\end{subequations}
and
\begin{subequations}
\label{eq:MDDCF:SE_sub2}
\begin{align}
\text{Sub 2:} \ \ \ &\max_{\pmb{p}^{\text{II}},\pmb{\lambda}^{\text{II}},\pmb{\mu}^{\text{II}}} \ \Lambda_{\text{SE}}^{\text{II}}  \\ 
\text{s.t.} \ &~~(\ref{eq:MDDCF:SE_sup}\text{b}).
\end{align}
\end{subequations}
Note that, since the above two sub-problems are strongly correlated, we introduce (\ref{eq:MDDCF:SE_sub1}c) in the first sub-problem to build the connection between them, where $\eta_l$ is a configurable factor for controlling the performance of channel estimation, which is elaborated in the following remark. 
\begin{remark}
The value of $\eta_l$ directly influences the accuracy of channel estimation, and further affects the DL/UL transmissions in Phase II. For instance, as $\eta_l$ decreases, to meet the constraint of (\ref{eq:MDDCF:SE_sub1}c), APs have to decrease their transmit power so as to reduce $I_l$ at the cost of compromised SE in Phase I. By contrast, a smaller $I_l$ results in the smaller estimation error in \eqref{eq:MDD_CF:ErrorMMSE}, which in turn results in a higher SE in Phase II. Hence, iteratively updating $\eta_l$ is essential for the proposed optimization problem. In our implementation, we resort to the bisection method to find the optimal $\eta_l$ in the range of $\left[0,\eta^{\text{max}}\right]$, where $0$ implies that the supplementary DL transmission in Phase I is not activated. In this case, there is no extra burden on the channel estimation in the current CT interval. By contrast, $\eta^{\text{max}}=\frac{\xi_l^{\text{SI}}I_l^{\text{max}}}{\sigma^2}$ denotes the maximum effective value, where $I_l^{\text{max}}$ means that all APs transmit at their highest possible power for the supplementary DL transmission in Phase I.
\end{remark}

Therefore, in the proposed algorithm, we first set $\eta_l=0$, and obtain the total initial SE of $\Lambda_{\text{SE}}^{\text{TPCT}(0)}=\Lambda_{\text{SE}}^{\text{I}(0)}+\Lambda_{\text{SE}}^{{\text{II}}(0)}$, where $\Lambda_{\text{SE}}^{\text{I}(0)}=0$.  Then, $\eta_l$ is iteratively updated based on the bisection method until convergence is achieved. The overall algorithm is summarized in Algorithm \ref{MDDCF:al3}.

\begin{algorithm}
\caption{SE Maximization within TPCT Interval} 
\label{MDDCF:al3}
\textbf{Initialization:} \\
Set $t=0$, $t^{'}=0$ $\eta_l \in [\eta^{a},\eta^{b}], \ \eta^{a}=\eta_l^{(0)}=0, \forall l, \eta^{b}=\eta^{\text{max}}$\;
Solve Sub 1 and Sub 2 problems in sequence using the QT method, and obtain $\Lambda_{\text{SE}}^{\text{I}(t)}$ and $\Lambda_{\text{SE}}^{{\text{II}}(t)}$\;
Compute $\Lambda_{\text{SE}}^{\text{TPCT}(t)}=\Lambda_{\text{SE}}^{{\text{I}}(t)}+\Lambda_{\text{SE}}^{\text{II}(t)}$\;
\LT{}{
\Repeat{$\textup{Convergence}$}{
Set $t = t + 1$, and update $\eta_l^{(t)}=\frac{\eta^{a}+\eta^{b}}{2}$\;
Implement Step 3 and Step 4\;
\eIf{\em{$\Lambda_{\text{SE}}^{\text{TPCT}(t)}<\Lambda_{\text{SE}}^{\text{TPCT}(t^{'})}$}}{$\eta^{b}=\eta_l^{(t)}$;}{$\eta^{a}=\eta_l^{(t)}$, $t^{'}=t$\;} 
}
}
\KwOut{$\pmb{p}^{\text{I}},\pmb{\lambda}^{\text{I}},\pmb{\mu}^{\text{I}},\pmb{p}^{\text{II}},\pmb{\lambda}^{\text{II}},\pmb{\mu}^{\text{II}},\Lambda_{\text{SE}}^{\text{TPCT}\ast}$}
\end{algorithm}

\section{Simulation Results and Discussion}\label{sec:MDD_CF:sim}

In this section, we provide the numerical results for comparison of MDD-, IBFD- and TDD-CF in distributed CF-mMIMO systems in terms of SE. All the simulations are implemented on MATLAB using the CVX tool \cite{cvx}.

\subsection{Parameters and Setup}
The following results are obtained based on either one CT interval or one radio frame consisting of one CT interval and ($N_{\text{c}}-1$) TPCT intervals. In our studies, we assume that the subcarrier spacing is 15 KHz (with the central carrier frequency of 5 GHz), and each OFDM symbol with cyclic prefix spans $t_{\text{s}}=71.35 \mu$s \cite{3gpp2017nr}. We assume that all MSs move at a relative speed of $v=5$ km/h, while APs are stationary. Hence, the coherence time is $t_{\text{ct}}\approx21.6$ ms. In this case, one CT interval can accommodate about $T_c=t_{\text{ct}}/t_{\text{s}}\approx 300$ OFDM symbols. Regarding the frame structure\cite{3gpp2017nr,access2009physical}, we assume that for all systems, the pilot transmission requires 15 OFDM symbol durations within one CT (i.e., $\gamma^{\text{P}}=15$, as shown in Fig. \ref{figure-MDDCF-Coh} and Fig. \ref{figure-MDDCF-MulCoh}), while for TDD systems, the GP lasts for 15 OFDM symbol durations in a TDD radio frame (i.e., $\gamma^{\text{G}}=15$). 

Assuming the delay spread of 40 ns \cite{zhao2002propagation}, the coherence bandwidth is $B_c\approx 4.2$ MHz. To make subcarrier signals experience flat fading, we assume that the total number of subcarriers is $M_{\text{sum}}=48$ for all systems, while in MDD systems, the numbers of DL and UL subcarriers are $M=32$ and $\bar{M}=16$, respectively. For the sake of fair comparison, in TDD systems, the ratio of DL/UL transmission times is set to the same ratio of DL/UL subcarrier numbers in MDD-CF, i.e., $\gamma_{\text{TDD}}^{\text{DL}}=180$ and $\gamma_{\text{TDD}}^{\text{UL}}=90$, if the CT interval has 300 symbols, while the other 30 symbols are pilot symbols and GP intervals. Then, the total SE of TDD-CF within one CT or one radio frame can be expressed as
\begin{align}
\Lambda_{\text{SE}}^{\text{TDD}}&=\frac{1}{M_{\text{sum}}}\sum_{d \in \mathcal{D}} \Big(\sum_{m=1}^{M_{\text{sum}}}\frac{\gamma_{\text{TDD}}^{\text{DL}}}{T_{\text{c}}}R(\text{SINR}_{d,m}) \nonumber \\
&+\sum_{\bar{m}=1}^{M_{\text{sum}}}\frac{\gamma_{\text{TDD}}^{\text{UL}}}{T_{\text{c}}}R(\text{SINR}_{d,\bar{m}})\Big).
\end{align}   

Owing to the DL transmission in Phase I, as shown in Figure \ref{figure-MDDCF-MulCoh}, the total SE achieved by the MDD-CF in one radio frame is different from that in one CT interval, which can be expressed as
\begin{equation}
\Lambda_{\text{SE}}^{\text{MDD-RF}}=\frac{1}{N_c}\Lambda_{\text{SE}}+\frac{N_c}{N_c-1}\Lambda_{\text{SE}}^{\text{TPCT}},
\end{equation}
where $\Lambda_{\text{SE}}$ and $\Lambda_{\text{SE}}^{\text{TPCT}}$ are derived based on \eqref{eq:MDDCF:MDDSE} and \eqref{eq:MDDCF:SE_sup}, respectively. The SE computation of IBFD-CF within one radio frame is similar to that of MDD-CF. 

We assume that APs in both IBFD-CF and MDD-CF schemes with distributed operation are capable of providing 30 dB IAI suppression in the propagation/analog domain with the existing approaches as mentioned previously. Provided that the 12-bit ADC is employed, MDD-CF can suppress IAI up to 72 dB (i.e., $\xi_l^{\text{MDD-IAI}}=-72$ dB, $\forall l$), of which 42 dB is attributed to the digital cancellation by FFT, as analyzed in Fig. \ref{figure-MDDCF-SICdia}.  On the contrary, as IBFD-CF can hardly cope with the IAI in digital domain, we assume that it can provide no more than 10 dB of digital-domain IAI mitigation\footnote{In fact, authors in \cite{nguyen2020spectral} presented several digital-domain methods to suppress IAI, which end up with only providing about 10 dB of IAI mitigation. The study implies that the IAI suppression in IBFD-CF scheme with centralized operation is very challenging, not to mention the CF systems operated in a distributed way.}, yielding $\xi_l^{\text{IBFD-IAI}}=-40$ dB, $\forall l$.  In the context of the IMI suppression, since the propagation/analog-domain IMI methods are relatively complicated to implement at the single-antenna MS, it is assumed that $\xi_d^{\text{IBFD-IMI}}=0$ dB, $\forall d$, while $\xi_d^{\text{MDD-IMI}}=-42$ dB, $\forall d$, owing to the FFT-assisted suppression. 

We assume that all APs and MSs are uniformly and randomly distributed within a square of size $\left(S_{\text{D}} \times S_{\text{D}}\right)$ $\text{m}^2$. The large-scale fading coefficients $\beta \in\left\{\beta_{ld}, \beta_{ll^{\prime}}, \beta_{dd^{\prime}}\right\}$ is given by $\beta[\text{dB}] = \text{PL} + \sigma_{\text{sh}}z$, where the shadowing is characterized by $\sigma_{\text{sh}}z$ with a standard deviation of $\sigma_{\text{sh}}=4$ dB and $z \sim \mathcal{N}(0,1)$. The PL exponent is assumed to be -3.8 \cite{demir2021foundations}. Unless otherwise noted, the other parameters are listed in Table \ref{Table:MDDCF:para}. 
\begin{table}
\caption{Simulation parameters}
\centering
\begin{tabular}{|l|l|}
\hline
Default parameters & Value  \\ \hline
Number of APs, MSs ($L,D$)  & (12, 4)  \\ \hline
Number of antennas per AP ($N$) & 6 \\ \hline
AP's and MS's Power budget ($P_l, P_d, \forall l,d$) & $(10,1) \ W$ \\ \hline
UL pilot power ($p_d^{\text{p}}, \forall d$) & 0.6 $W$ \\ \hline
MSs' QoS requirements  ($\chi_{\text{DL}}, \chi_{\text{UL}}$) & (0.5, 0.1)  \\ \hline
Noise power ($\sigma^2$)  & -94 dBm   \\ \hline
Delay taps ($U$) & 4 \\ \hline 
Number of CT intervals within one frame ($N_c$) & $10$  \\ \hline
Residual SI level at AP ($\xi_l^{\text{MDD-SI}},\xi_l^{\text{IBFD-SI}}, \forall l$) &-130 dB\\ \hline
Residual SI level at MS ($\xi_d^{\text{MDD-SI}},\xi_d^{\text{IBFD-SI}}, \forall d$) & -120 dB \\ \hline 
Cell length ($S_{\text{D}}$) & 400 m \\ \hline
\end{tabular}
\label{Table:MDDCF:para}
\end{table}

\subsection{Case of One CT Interval}

\begin{figure}
\centering
\includegraphics[width=0.8\linewidth]{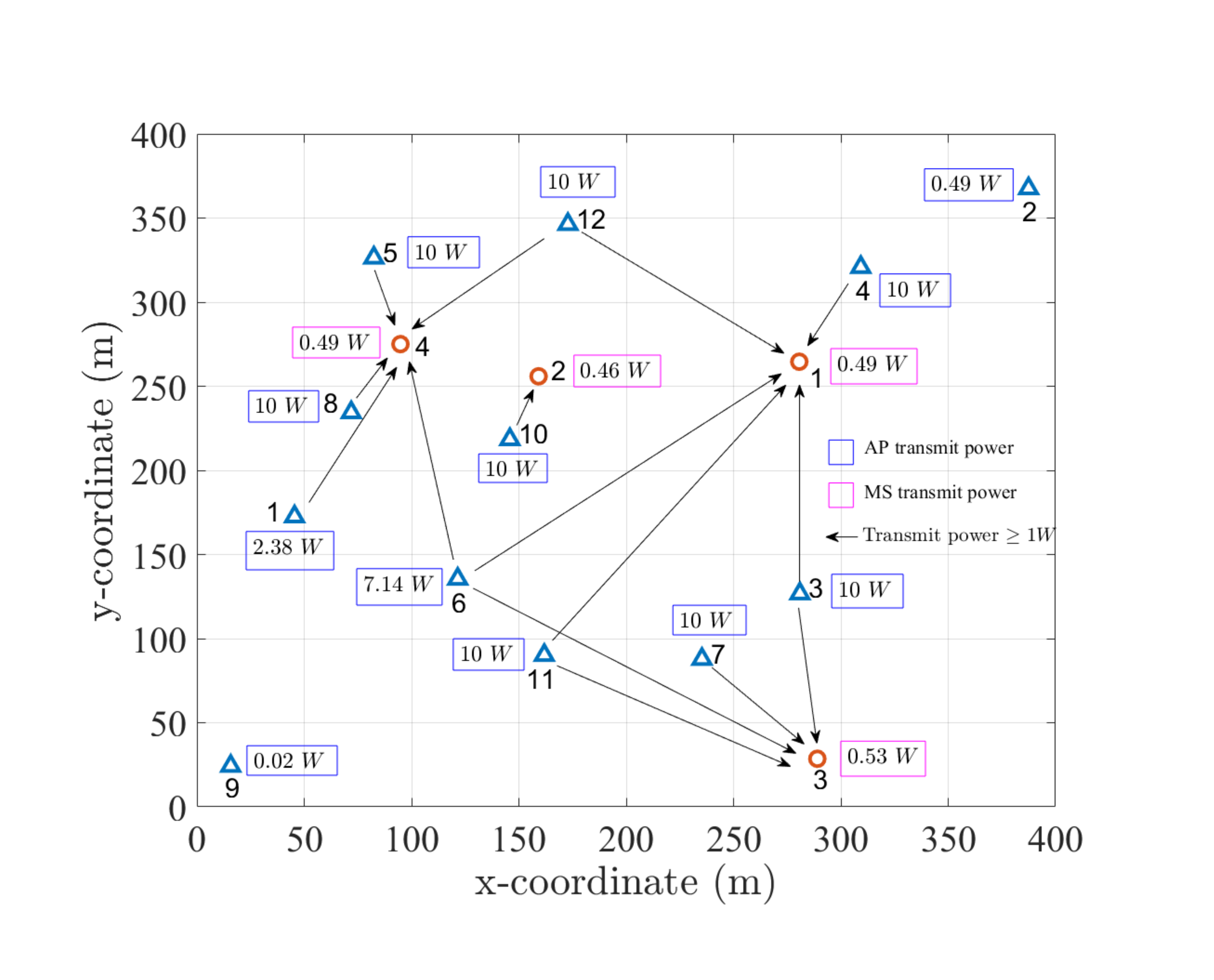}
\caption{AP-selection and power-allocation results obtained by Algorithm 1 in MDD-CF scheme.}
\label{figure-MDDCF-figsup}
\end{figure}

\begin{figure}
\centering
\includegraphics[width=0.8\linewidth]{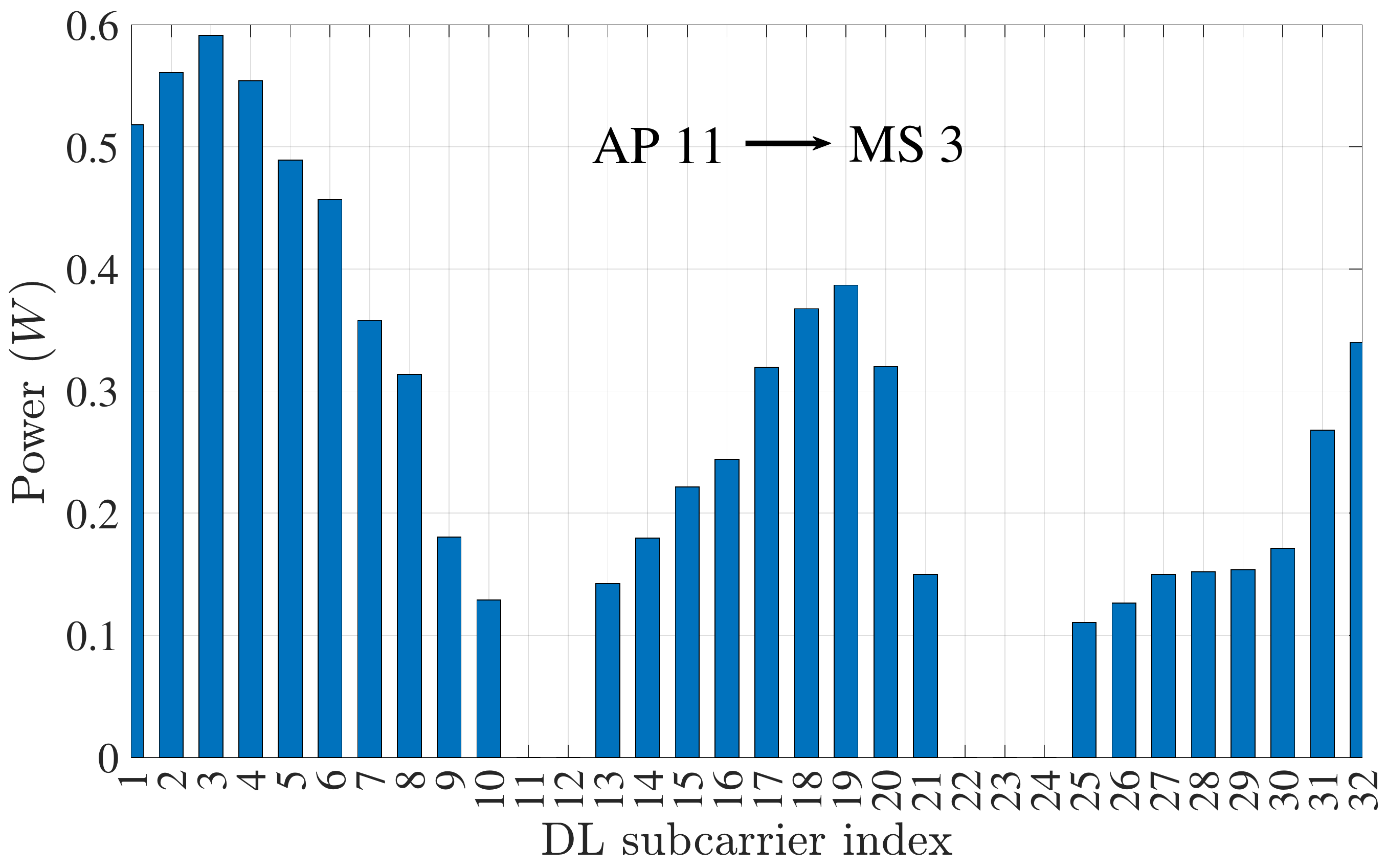}
\caption{Power- and subcarrier-allocation results of AP 11 to MS 3 achieved by Algorithm 1.}
\label{figure-MDDCF-figsup2}
\end{figure}

\begin{figure}
\centering
\includegraphics[width=0.7\linewidth]{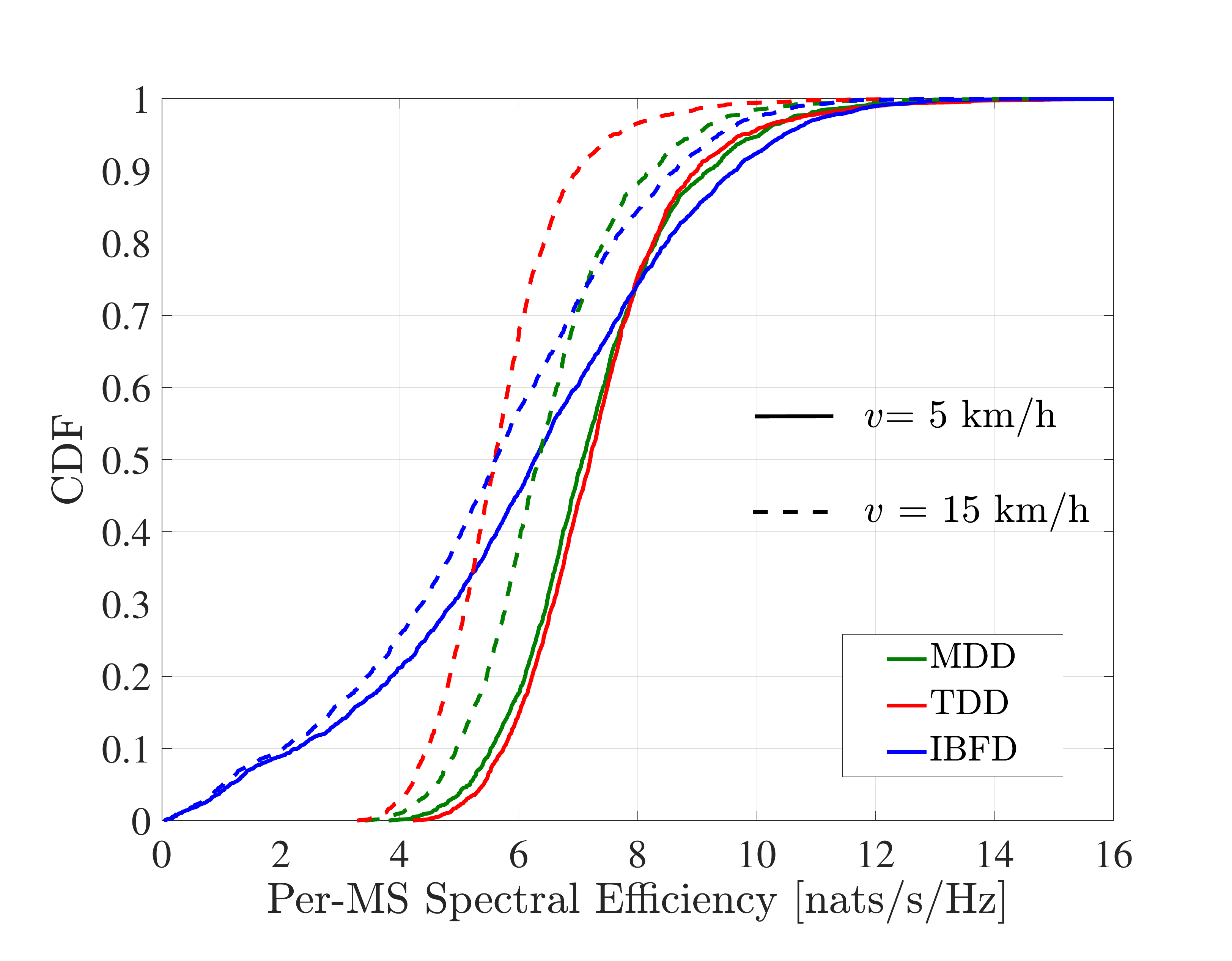}
\caption{Cumulative distribution of the per-MS in different CF schemes.}
\label{figure-MDDCF-fig1}
\end{figure}

Given a randomly generated network layout, Fig.~\ref{figure-MDDCF-figsup} shows the AP-selection and power-allocation results attained by Algorithm~\ref{MDDCF:al1} over one CT interval. The numbers within the blue and pink squares denote the AP and MS's transmit power, respectively. The black arrowed lines denote the DL links with the transmit power larger than $1~W$. It can be observed from the figure that, subject to the limited SIC capability, the transmit power of a MS for UL transmission is much less than the budget power. Moreover, some APs, such as AP~$1$, AP~$2$ and AP~$9$, located far away from MSs are controlled by the algorithm to reduce their transmit power to avoid IAI on the other APs. Furthermore, in Fig.~\ref{figure-MDDCF-figsup2}, as an example, we show the detailed power- and subcarrier-allocation results of AP 11 to MS 3. Explicitly, based on \eqref{eq:MDD_CF:ldmpower}, AP 11 allocates different power to the 32 DL subcarriers, where some subcarriers, namely subcarriers 11, 12, 22, 23 and 24,  with very small power are not assigned to MS 3.

In the following simulations, we focus on the performance comparison of MDD-, TDD- and IBFD-CF schemes over one CT interval, as shown in Fig. \ref{figure-MDDCF-Coh}. First, the performance of three schemes is presented in Fig. \ref{figure-MDDCF-fig1}. From the results, when $v=5$ km/h (corresponding to $T_c=300$ OFDM symbols), TDD-CF slightly outperforms MDD-CF in terms of the 95\%-likely per-MS SE as the result of the IAI and IMI being not perfectly canceled in MDD-CF. Although IBFD-CF provides substantially higher SE for the strongest MSs, its 95\%-likely performance is only about 0.5 nats/s/Hz, which is 10 times lower than MDD- and TDD-CF (5 nats/s/Hz). The rationale is that there are usually only a small number of MSs that can benefit from IBFD mode. These MSs should be located far away from the neighboring MSs and their serving APs are also not close to each other. Otherwise, to meet the QoS requirements, see (\ref{eq:MDDCF:SE_formulation}g) and (\ref{eq:MDDCF:SE_formulation}h), the DL/UL transmit power has to be decreased so as to reduce the negative effect of IAI/IMI. Furthermore, the performance of these schemes is very different when the relative speed between MSs and APs is increased to 15 km/h, which corresponds to 100 OFDM symbols transmitted in one CT interval. Fig. \ref{figure-MDDCF-fig1} clearly shows that the performance of TDD-CF deteriorates quickly as the relative speed increases. This is because the length of GP is constant, but its proportion within one CT interval becomes larger when the relative speed goes up, which therefore leads to the reduced SE. 


\begin{figure}[]
\centering
\subfigure[]{
\includegraphics[width=0.7\linewidth]{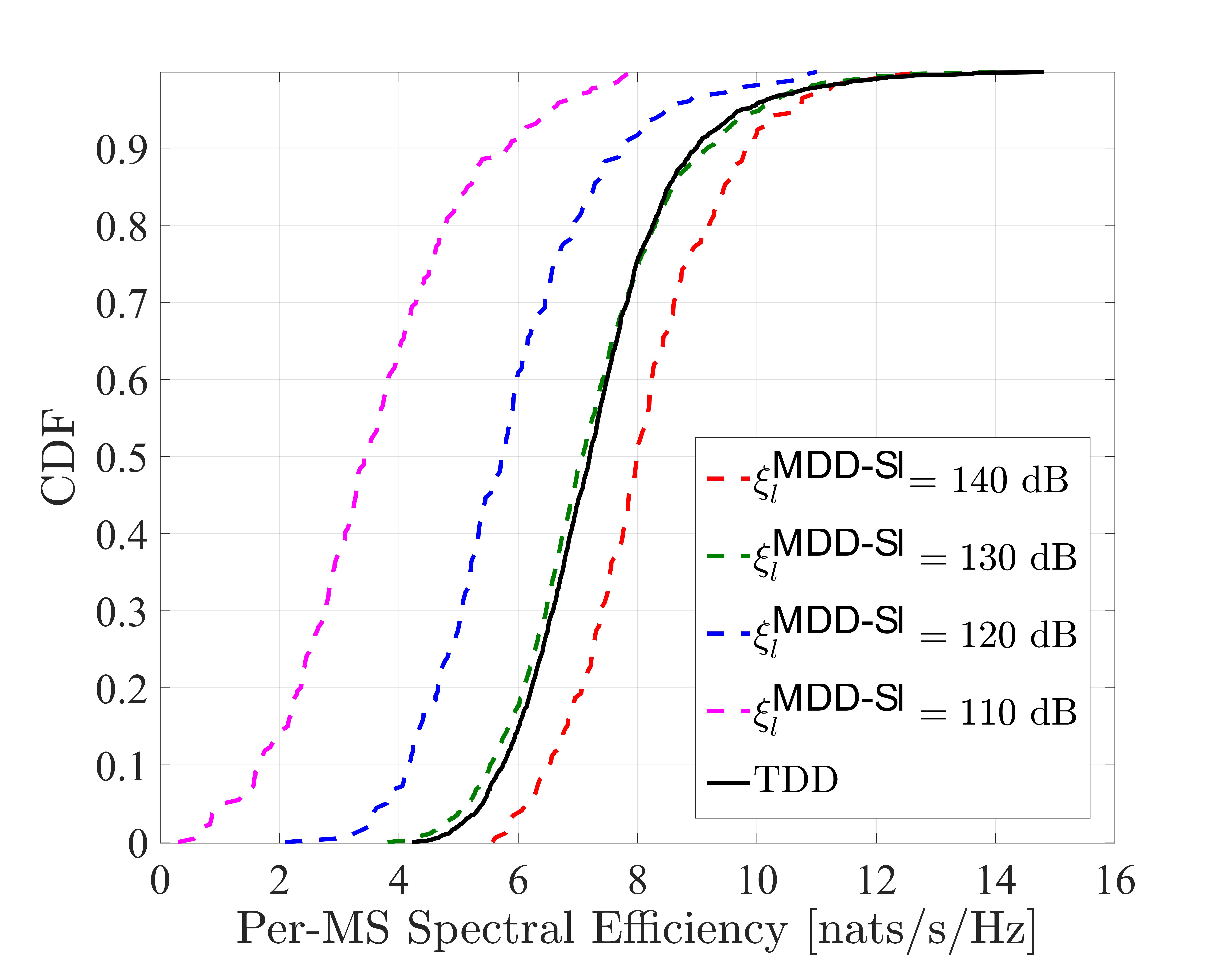}
}
\subfigure[]{
\includegraphics[width=0.7\linewidth]{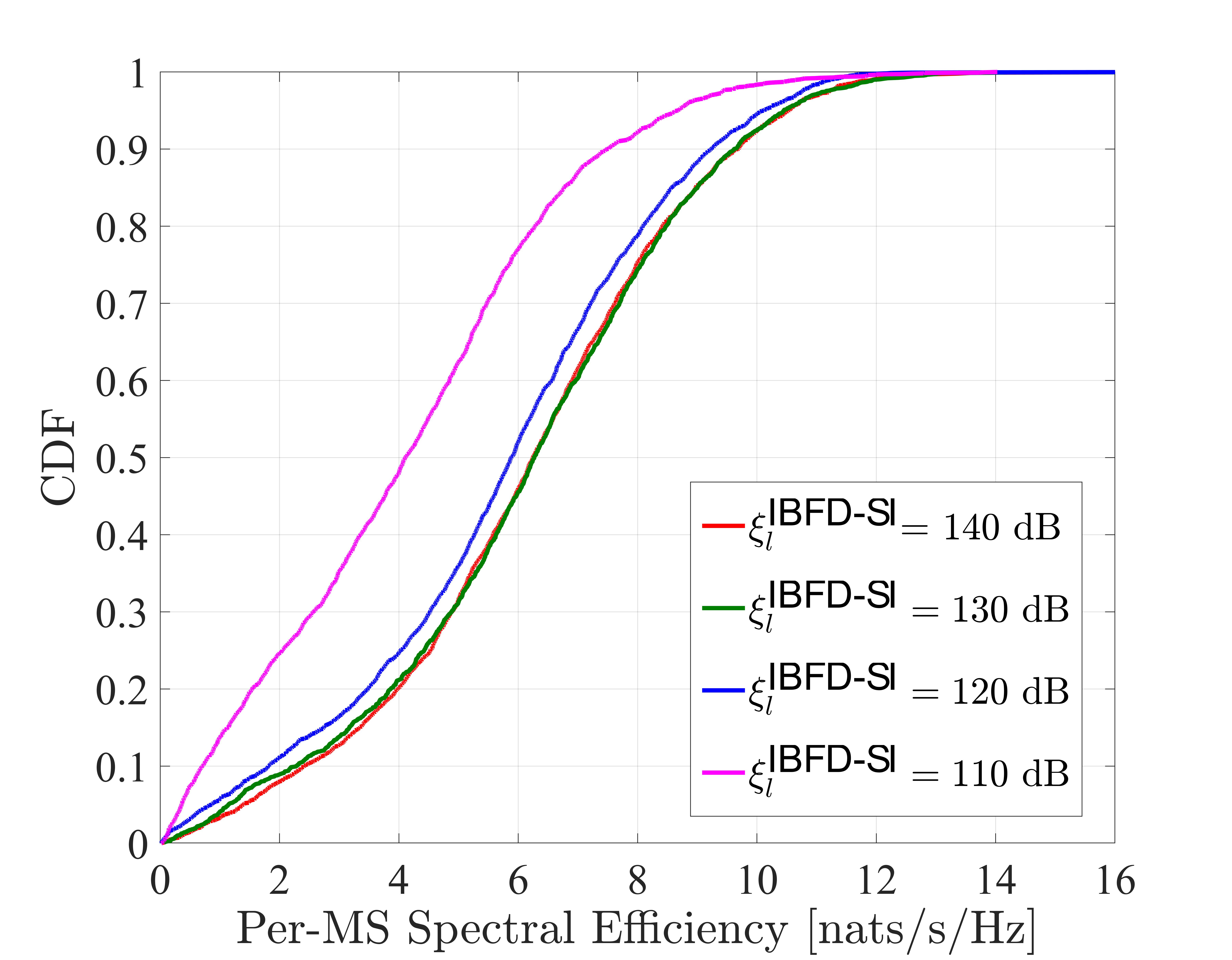}
}
\caption{Cumulative distribution versus Per-MS SE in MDD- and IBFD-CF schemes with different SIC.}
\label{figure-MDDCF-fig3}
\end{figure}

The influence of the SIC capability at APs and MSs in the MDD- and IBFD-CF schemes is demonstrated in Fig. \ref{figure-MDDCF-fig3}. Note that the value of $\xi_d^{\text{MDD-SI}}$ is not shown in the figures, which is always set to be 10 dB lower than that of $\xi_l^{\text{MDD-SI}}$. From Fig. \ref{figure-MDDCF-fig3}(a), it can be seen that the higher SIC capability the APs and MSs have, the better performance the MDD-CF can attain. Specifically, when $\xi_l^{\text{MDD-SI}}=-140$ dB, MDD-CF outperforms TDD-CF with a higher performance upper bound. This is because MDD mode is free of GPs, and hence have more time resource for data transmissions. 
On the contrary, from Fig. \ref{figure-MDDCF-fig3}(b), the interesting result is that when the $\xi_l^{\text{IBFD-SI}}$ reaches at 120 dB, the further increased SIC capability can hardly improve the performance of IBFD-CF. This implies that the IBFD-CF is mainly IAI/IMI-limited. 

Fig. \ref{figure-MDDCF-fig2} evaluates the influence of cell size on the performance of MDD- and IBFD-CF schemes. It can be clearly seen that as the cell size decreases, meaning that the distribution of MS and AP become denser, the 95\%-likely per-MS SE of the MDD-CF is increased by 4 nats/s/Hz, while the IBFD-CF fails to obtain the significant gain from the denser network deployment. The reason is that, in IBFD-CF, the denser distribution of APs and MSs also means the shorter interference links of AP-AP and MS-MS. In this case, the large-scale fading is unable to provide enough IAI and IMI mitigation, consequently, both the APs and MSs have to decrease their transmit power so as to control the level of interferences. By contrast, the FFT-assisted IAI and IMI cancellation in MDD-CF can efficiently mitigate the interference in digital domain. Hence, MDD-CF can benefit significantly from the short-distance communications. Moreover, to further illustrate the effect of IAI, we assume that the APs in both MDD-CF and IBFD-CF schemes have an extra 30 dB of IAI suppression capability, when $S_{\text{D}}=100$ m. Then, it can be seen that the IBFD-CF has a significant improvement in terms of both the 95\%-likely and the median per-MS SE. However, the 95\%-likely per-MS SE of the IBFD-CF with extra IAI suppression is still 4 nats/s/Hz lower than that of the MDD-CF without extra IAI suppression, due to the less IMI mitigation. We can conclude from Fig. \ref{figure-MDDCF-fig3} and Fig. \ref{figure-MDDCF-fig2} that the IBFD is not a desirable mode for FD-style operation in distributed CF-mMIMO systems, while the MDD mode is more promising.  

\begin{figure}
\centering
\includegraphics[width=0.7\linewidth]{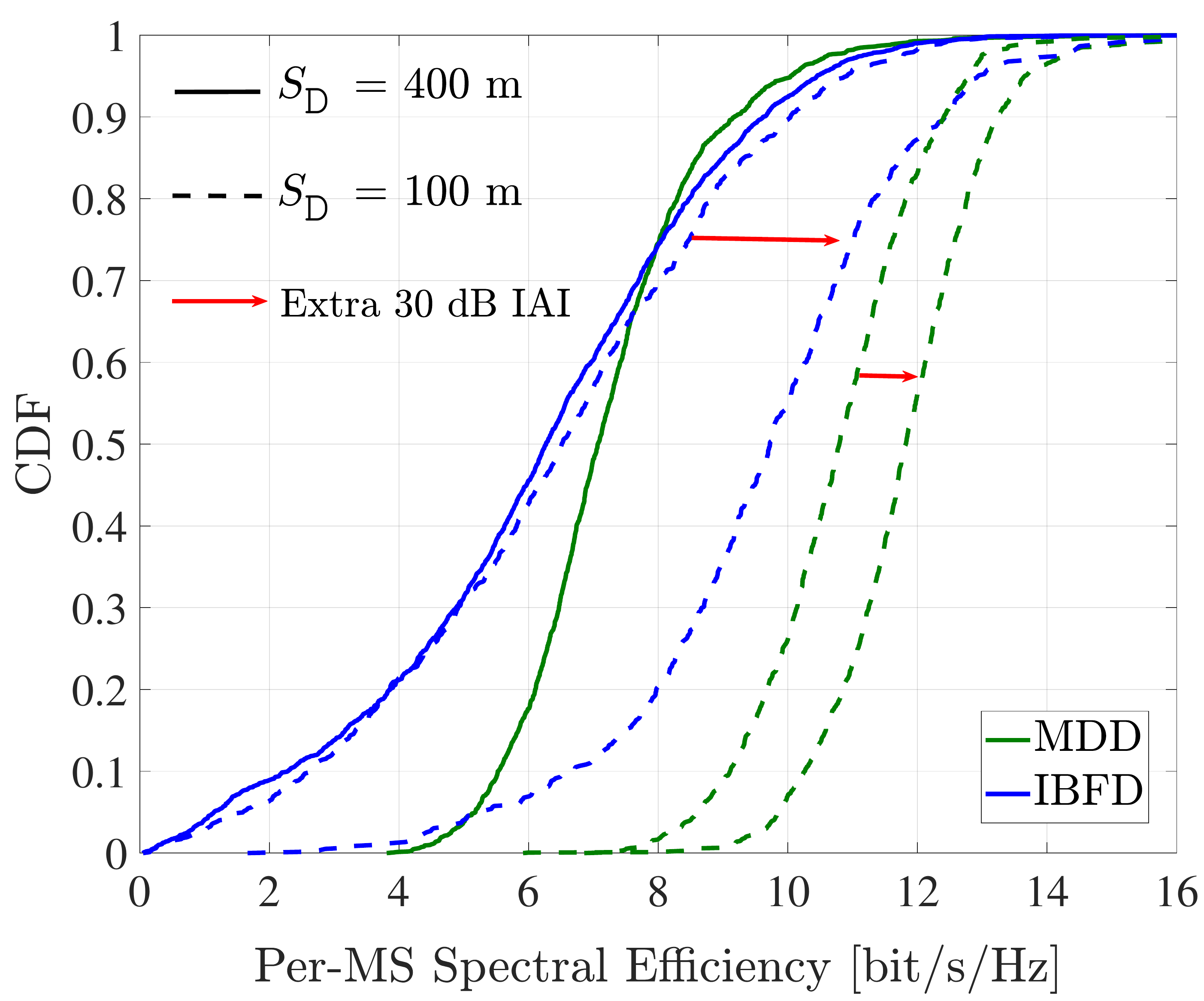}
\caption{Cumulative distribution versus Per-MS SE in MDD- and IBFD-CF schemes with different cell size.}
\label{figure-MDDCF-fig2}
\end{figure}

The convergence behavior of the Algorithm \ref{MDDCF:al1} in MDD-CF is shown in Fig. \ref{figure-MDDCF-fig6}, where each iteration period denotes one cycle, which includes step 3 - step 13 in the Algorithm. The lowest starting point at the beginning of each iteration period is due to the reinitialization at step 4. From the figure, it can be easily seen that Algorithm \ref{MDDCF:al1} can reach convergence only after two iteration periods. Concerning the complexity, Algorithm \ref{MDDCF:al1} is mainly attributed to solving the optimization problem \eqref{eq:MDDCF:QT}, regardless of the computation of the ZF matrices and the initialization of the involved variables. Hence, according to \cite{peaucelle2002user}, the approximated complexity of Algorithm \ref{MDDCF:al1} is $\mathcal{O}\big((LDM+2DM+3D\bar{M})^2(L+D+3DM_{\text{sum}})^{2.5}+(L+D+3DM_{\text{sum}})^{3.5}\big)$ per iteration. From these results, one can see that the numbers of MSs and subcarriers are the two dominant factors limiting the scalability of MDD-CF, where the MSs are densely distributed with a large number of available subcarriers. To this end, the more scalable approaches for the MDD-CF scheme deserve further investigation in the future.

\begin{figure}
\centering
\includegraphics[width=0.7\linewidth]{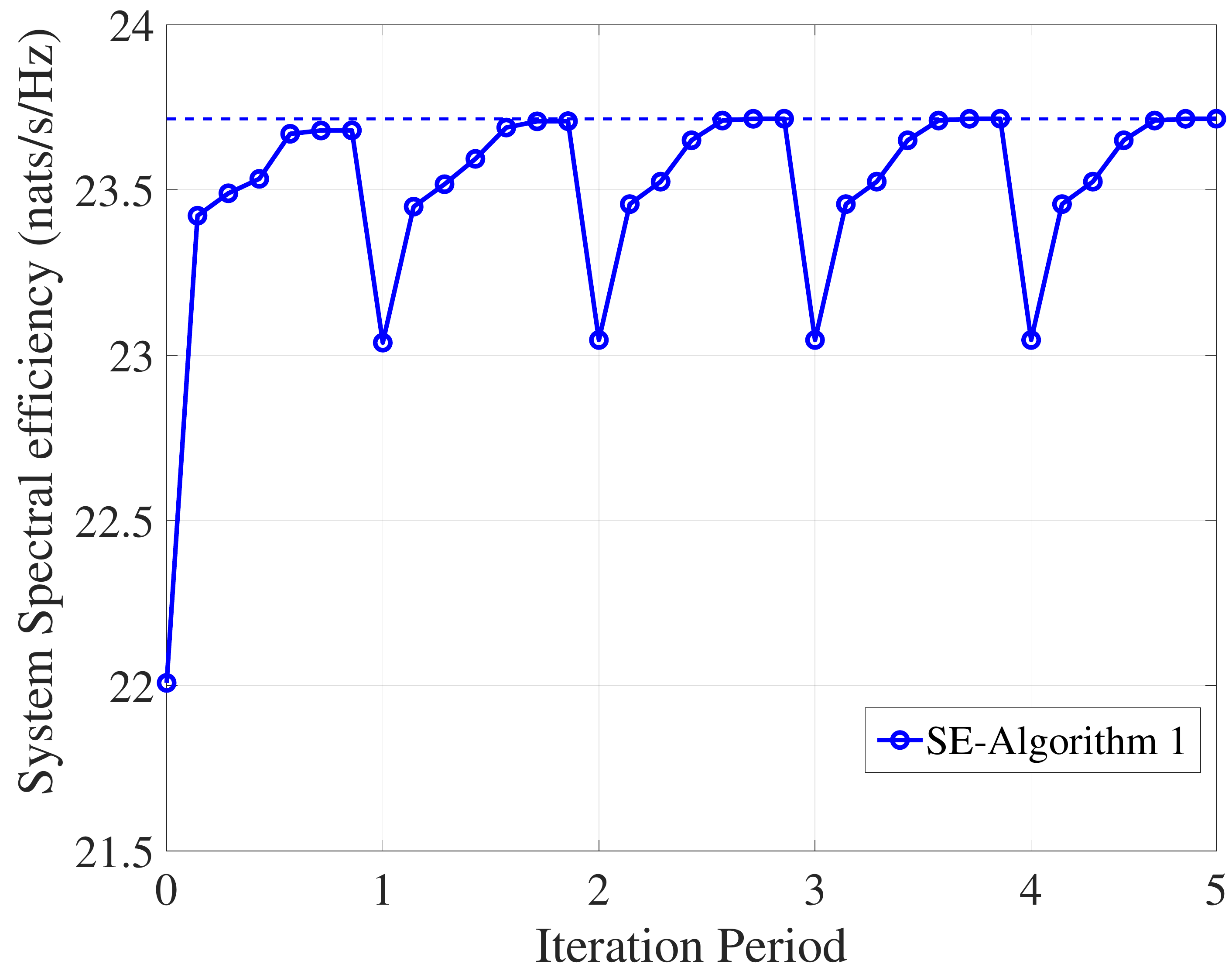}
\caption{SE convergence behavior of Algorithm 1 in MDD-CF scheme.}
\label{figure-MDDCF-fig6}
\end{figure}

\subsection{Case of One Radio Frame}


In this subsection, we demonstrate the advantages of MDD-CF, which can fully exploit the time resource within one radio frame with the aid of Algorithm \ref{MDDCF:al3} to improve the SE performance, as detailed below.


\begin{figure}
\centering
\includegraphics[width=0.7\linewidth]{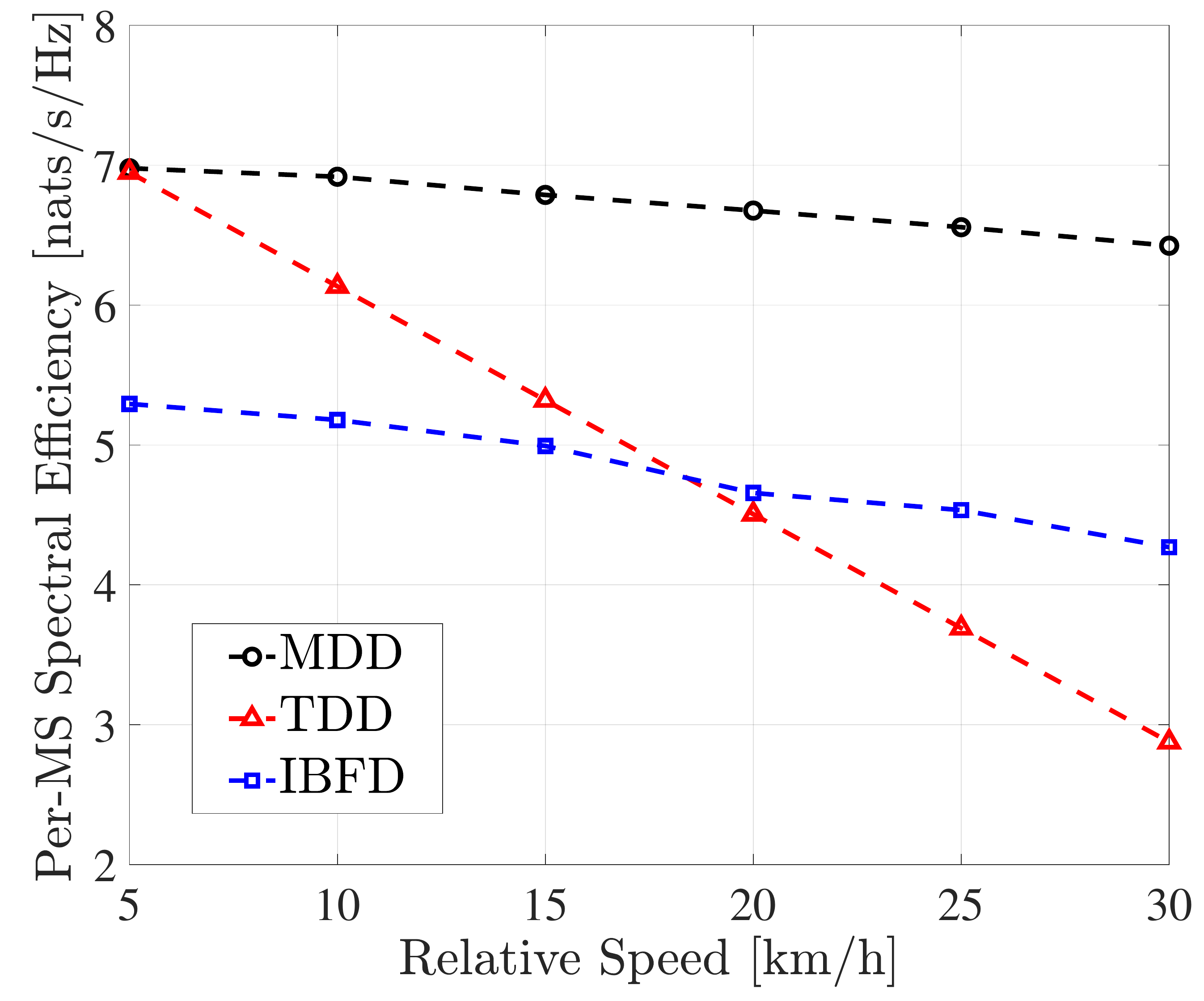}
\caption{Per-MS SE of one radio frame versus the relative speed.}
\label{figure-MDDCF-fig5}
\end{figure}

Fig. \ref{figure-MDDCF-fig5} evaluates the per-MS SE of the three different schemes over a range of relative speeds. Explicitly, the SE performance of TDD-CF deteriorates quickly, as the relative speed increases. The explanation is as follows. In TDD mode, DL/UL transmissions can only be implemented in a sequential manner. When the relative speed between MSs and APs increases, the reduced CT interval leads to the less time for data transmission. By contrast, with the aid of our proposed Algorithm \ref{MDDCF:al3}, both MDD-CF and IBFD-CF are robust to the high-mobility scenarios, as expected. In particular, as MDD-CF can benefit from the lower IAI/IMI as well as the FD operation during both Phase I and II, it doubling the per-MS SE compared to TDD-CF, at the speed of 30 km/h. As shown in Fig. \ref{figure-MDDCF-fig5}, the IBFD-CF is also outperformed by the MDD-CF, due to the larger residual IAI/IMI, which not only affect the UL and DL communications in Phase II, but also the channel estimation in Phase I. However, it can eventually surpass the TDD-CF at about 18 km/h and achieve 1 nats/s/Hz higher per-MS SE than the TDD-CF at 30 km/h, due to the fact that Phase I becomes increasingly paramount in TPCT interval.

\begin{figure}
\centering
\includegraphics[width=0.7\linewidth]{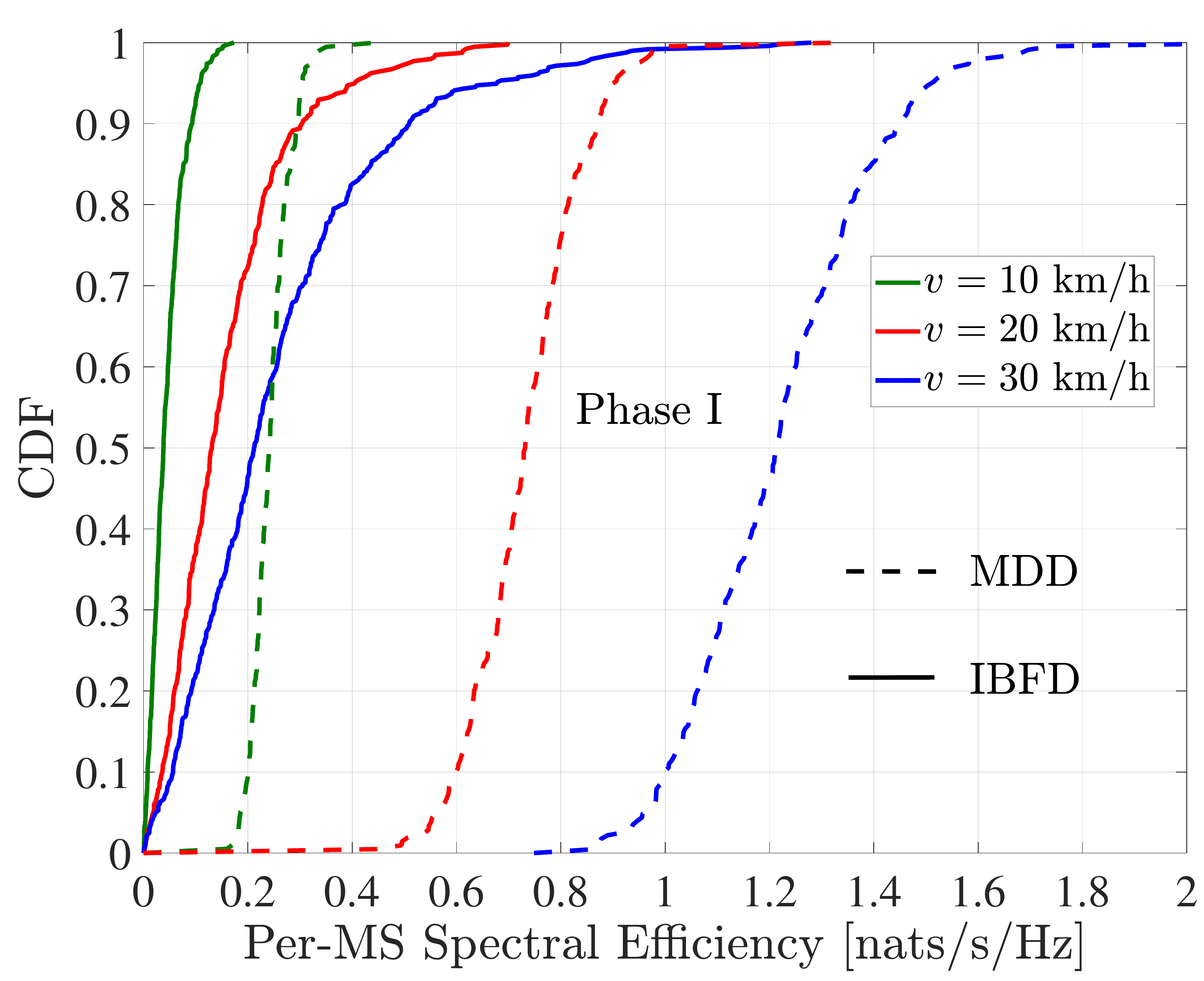}
\caption{Cumulative distribution versus per-MS SE in Phase I.}
\label{figure-MDDCF-fig5_1}
\end{figure}

\begin{figure}
\centering
\includegraphics[width=0.7\linewidth]{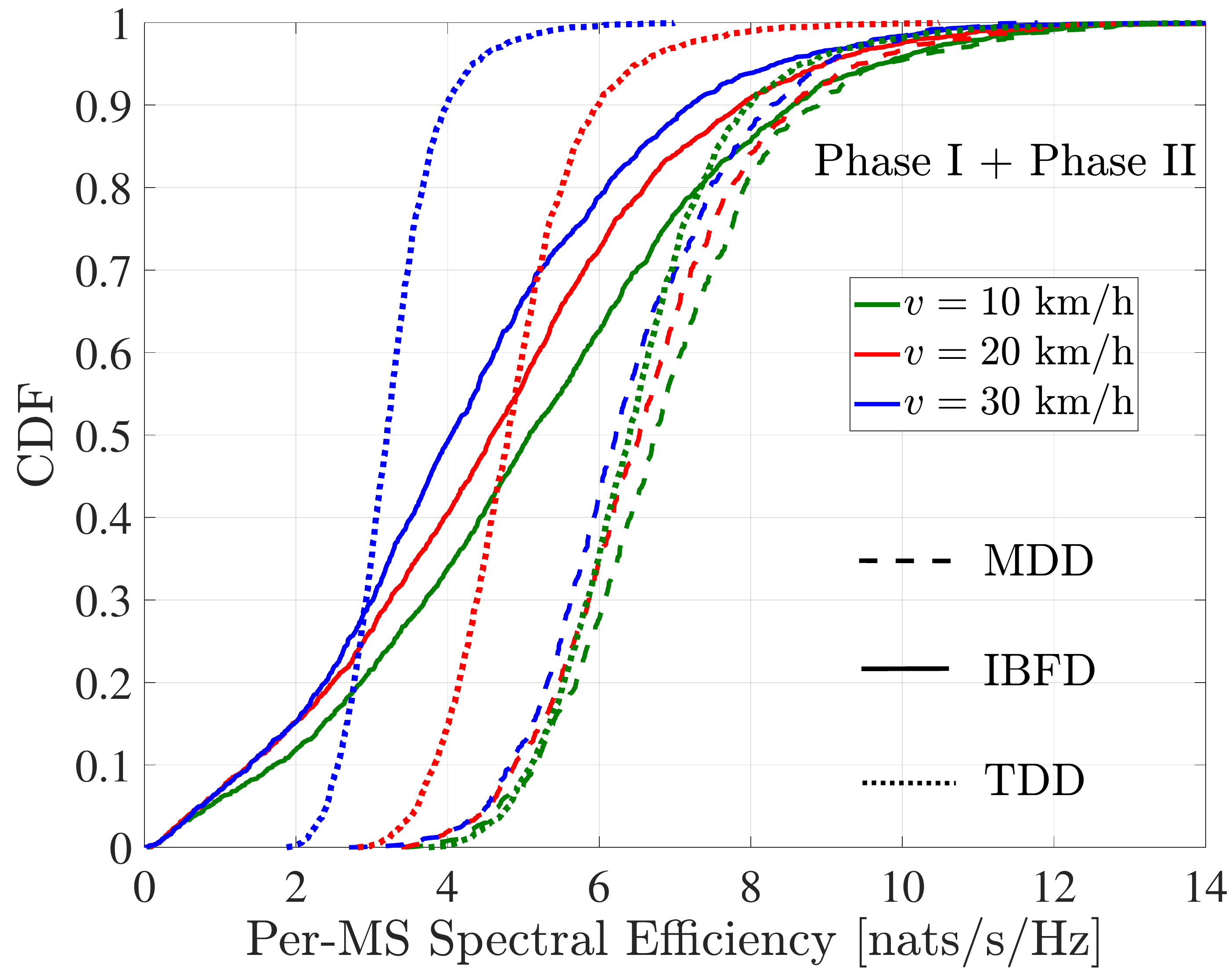}
\caption{Cumulative distribution versus per-MS SE in the case of combined Phase I and Phase II.}
\label{figure-MDDCF-fig5_2}
\end{figure}

The performance of Phase I and that of the combination of Phase I and Phase II are demonstrated in Fig. \ref{figure-MDDCF-fig5_1} and Fig. \ref{figure-MDDCF-fig5_2}, respectively. Due to the limited capability of IAI/IMI management in IBFD-CF, in order to guarantee the accuracy of channel estimation in Phase I, the APs have to cut down the DL's transmit power so as to limit the interference imposed on the pilot receiving. Consequently, as seen in Fig. \ref{figure-MDDCF-fig5_1}, the IBFD-CF attains a much poorer performance in Phase I than the MDD-CF. Furthermore, Fig. \ref{figure-MDDCF-fig5_2} shows that when the relative speed is increased to 30 km/h, the 95\%-likely per-MS SE of the MDD-CF is nearly two times higher than that of the TDD-CF. Additionally, as seen in Fig. \ref{figure-MDDCF-fig5_2}, although the IBFD-CF lags behind the TDD-CF in terms of the 95\%-likely per-MS SE, its median performance is higher than that of the TDD-CF at 30 km/h. 


\begin{figure}
\centering
\includegraphics[width=0.7\linewidth]{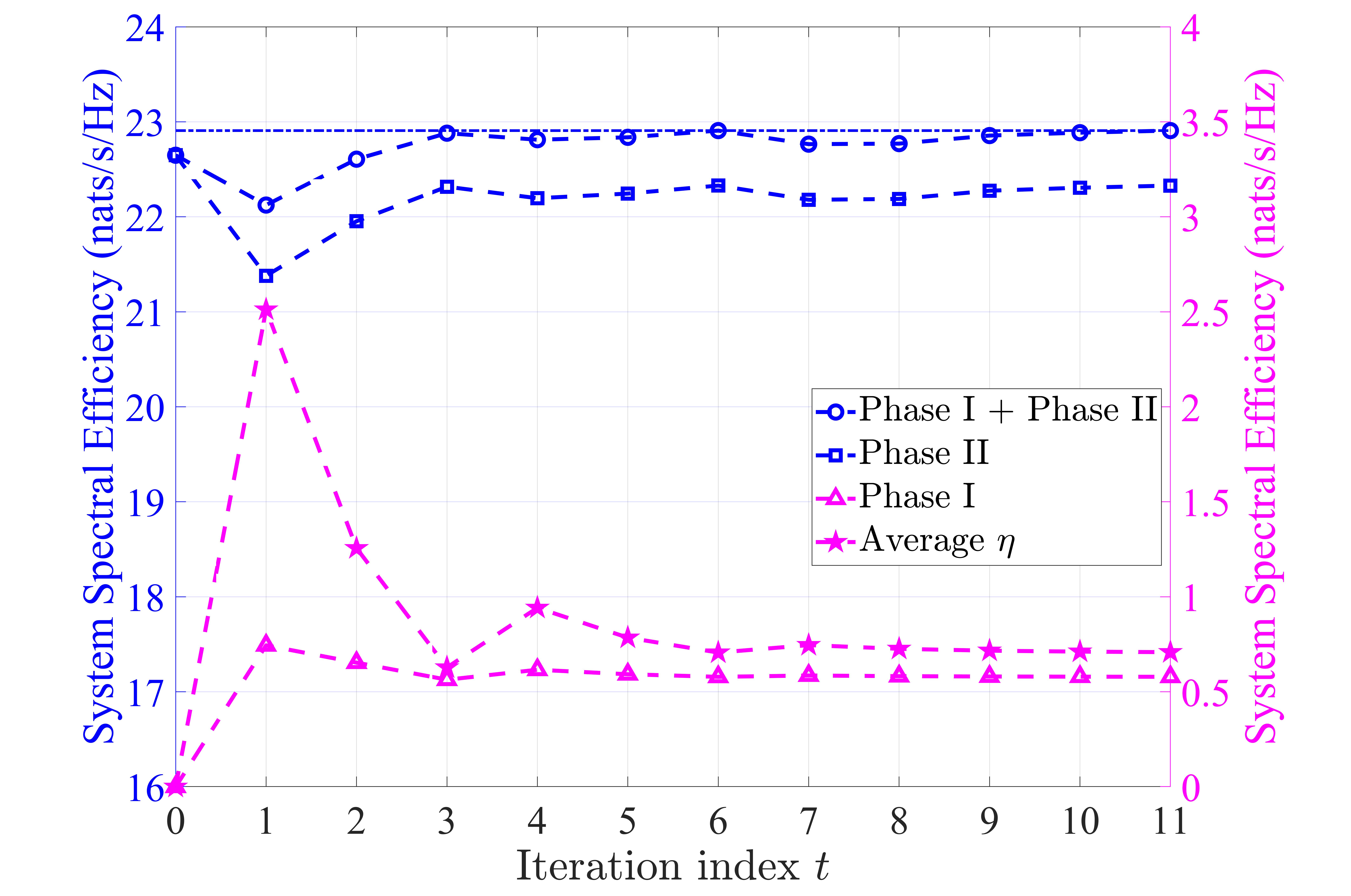}
\caption{SE convergence behavior of Algorithm 3 in MDD-CF scheme.}
\label{figure-MDDCF-fig7}
\end{figure}

Finally, the convergence of Algorithm \ref{MDDCF:al3} is shown in Fig. \ref{figure-MDDCF-fig7}. Note that $\eta$ marked in the figure denotes the average value of $\eta_l$, i.e., $\eta = \sum{\eta_l}/L$. It can be seen from the results that when $t=0$ and $\eta=0$, only Phase II is activated and there is no data transmission in Phase I. Then, as $\eta$ iterates in line 7 of Algorithm \ref{MDDCF:al3}, the general observation is that if the performance of one phase improves, the performance of the other phase degrades. The reason is that according to (\ref{eq:MDDCF:SE_sub1}c), if more power is allocated for the DL transmission in Phase I, it results in not only the SE increase but also a larger channel estimation error, which leads to the SE degradation in Phase II, and vice versa. The results demonstrate that our proposed algorithm is capable of attaining the 99\% system performance after 6 iterations, and fully converging within about 10 iterations. The complexity of Algorithm \ref{MDDCF:al3} is mainly attributed by the optimization of \eqref{eq:MDDCF:SE_sup}, which has two phases. However, as indicated by \eqref{eq:MDDCF:SE_sub1} and \eqref{eq:MDDCF:SE_sub2}, the complexity of Phase I is much lower than that of Phase II, as there are less variables to consider and also less constraints involved in the optimization of Phase I. Consequently, the complexity of Algorithm \ref{MDDCF:al3} is dominated by the operations in Phase II, which is nearly the same as Algorithm \ref{MDDCF:al1}.

\section{Conclusions}\label{sec:MDD_CF:concl}
The MDD-CF scheme has been proposed and comprehensively studied in terms of SE performance, when AP-selection, power- and subcarrier-allocation are taken into consideration. Firstly, the SE optimization problem has been studied in one single CT interval. We have exploited the interdependence of the involved variables and transformed the mixed-integer optimization to a continuous-integer convex-concave problem. In order to efficiently solve the problem, we have proposed a QT-assisted iterative algorithm, which is capable of quickly achieving a local optimum. The simulation results show that in distributed CF-mMIMO systems, MDD-CF can significantly outperform IBFD-CF due to the more efficient IAI and IMI mitigation. Secondly, the case of one radio frame with imperfect channel estimation has been studied. To fully take the advantages of FD operation in radio frame, we have designed a TPCT interval comprised of two tightly coupled phases. Then, a two-step iterative algorithm based on bisection method has been proposed to maximize the SE in TPCT interval. According to the simulation results, with the aid of TPCT interval and Algorithm \ref{MDDCF:al3}, MDD-CF is more robust to high-mobility scenarios, while the performance of TDD-CF degrades quickly with the increase of the relative speed between APs and MSs.

\appendices

\section{The simplification of $\text{SINR}_{d,m}$ and $\text{SINR}_{d,\bar{m}}$ in \eqref{eq:MDDCF:SINRSim}}\label{Appen:MDD_CF_SINRSim}
For $\text{SINR}_{d,m}$, since the ZF precoder is employed with the constraint of power normalization, the numerator of \eqref{eq:MDD_CF:dlSINR} can be transformed to $\left|\sum_{l \in \mathcal{L}}\lambda_{ld}\mu_{ldm} \sqrt{p_{ldm}}\omega_{ldm}\right|^2$, where $\omega_{ldm}=\frac{1}{\left\|\pmb{f}_{ld}^{\text{ZF}}[m]\right\|_2}$. In the denominator, $\text{MUI}_{d,m}\approx 0$ due to the ZF precoding, when $N\geq D$. Furthermore, $\text{var}\left\{z^{\text{SI}}_d\right\}$ and $\text{var}\left\{z^{\text{IMI}}_d\right\}$ can be obtained as follows
\begin{align}
\text{var}\left\{z^{\text{SI}}_d\right\}&=\mathbb{E}\left[\bar{z}^{\text{SI}}_d\left(\bar{z}^{\text{SI}}_d\right)^H\right]\overset{(a)}{=}\xi_d^{\text{SI}}\sum_{\bar{m} \in \bar{\mathcal{M}}}\mu_{d\bar{m}}p_{d\bar{m}}, \nonumber \\
\text{var}\left\{z^{\text{IMI}}_d\right\}&=\xi_d^{\text{IMI}}\mathbb{E}\left[\bar{z}^{\text{IMI}}_d\left(\bar{z}^{\text{IMI}}_d\right)^H\right] \nonumber \\
& \overset{(b)}{=}\xi_d^{\text{IMI}}\sum_{d^{\prime} \in \mathcal{D}\backslash \left\{d\right\}}\sum_{\bar{m} \in \bar{\mathcal{M}}}\frac{\beta_{dd^{\prime}}}{M_{\text{sum}}}\mu_{d^{\prime}\bar{m}}p_{d^{\prime}\bar{m}},
\end{align}
where $(a)$ is derived using $\mathbb{E}\left[h_{dd}h_{dd}^H\right]=\xi_d^{\text{SI}}$ according to \eqref{eq:MDDCF:SICh}, and $(b)$ is obtained using $\mathbb{E}\left[h_{dd^{\prime}}[\bar{m}]h_{dd^{\prime}}^H[\bar{m}]\right]=\mathbb{E}\left[(\pmb{\phi}_{\text{UL}}^T\pmb{F}\pmb{\Psi}\pmb{g}_{dd^{\prime}})(\pmb{\phi}_{\text{UL}}^T\pmb{F}\pmb{\Psi}\pmb{g}_{dd^{\prime}})^H\right]=\frac{\beta_{dd^{\prime}}}{M_{\text{sum}}}$.

For the $\text{SINR}_{d,\bar{m}}$, since $(\pmb{w}_{ld}^{\text{ZF}}[\bar{m}])^H\pmb{h}_{ld}^{\text{ZF}}[\bar{m}]\approx 1$, the term in the numerator of \eqref{eq:MDD_CF:ulSINR} can be changed to $\mu_{d\bar{m}}p_{d\bar{m}}L^2$. In the denominator, $\text{MUI}_{d,\bar{m}}\approx 0$, while the second term can be obtained as 
\begin{align}
&\sum_{l \in \mathcal{L}}\mathbb{E}\left[\left\|\left(\pmb{w}_{ld}^{\text{ZF}}[\bar{m}]\right)^H\pmb{z}^{\text{SI}}_l\right\|^2\right] \nonumber \\
&= \sum_{l \in \mathcal{L}}\text{Tr}\left[\left(\pmb{w}_{ld}^{\text{ZF}}[\bar{m}]\right)^H\text{diag}\left(\text{cov}\left\{\bar{\pmb{z}}^{\text{SI}}_l\right\}\right)\pmb{w}_{ld}^{\text{ZF}}[\bar{m}]\right]\nonumber \\
&\overset{(a)}{=}\xi_l^{\text{SI}}\sum_{l \in \mathcal{L}}\left\|\pmb{w}_{ld}^{\text{ZF}}[\bar{m}]\right\|_2^2\sum_{m \in \mathcal{M}}\sum_{d \in \mathcal{D}} \lambda_{ld}\mu_{ldm} p_{ldm},
\end{align}
where $(a)$ holds since
\begin{align}
&\left(\text{cov}\left\{\bar{\pmb{z}}^{\text{SI}}_l\right\}\right)_{i,i} \nonumber \\
&= \sum_{m \in \mathcal{M}}\sum_{d \in \mathcal{D}} \lambda_{ld}\mu_{ldm} p_{ldm}\mathbb{E}\left[\left|\pmb{H}_{ll}^{(i,:)}\pmb{f}_{ld}^{\text{ZF}}[m] x_d[m]\right|^2\right] \nonumber \\
&\overset{(b)}{=}\sum_{m \in \mathcal{M}}\sum_{d \in \mathcal{D}} \lambda_{ld}\mu_{ldm} p_{ldm}\times \nonumber \\
&\text{Tr}\left\{\pmb{f}_{ld}^{\text{ZF}}[m]\left(\pmb{f}_{ld}^{\text{ZF}}[m]\right)^H \mathbb{E}\left[\left(\pmb{H}_{ll}^{(i,:)}\right)^H\pmb{H}_{ll}^{(i,:)}\right]\right\} \nonumber \\
&\overset{(c)}{=}\xi_l^{\text{SI}}\sum_{m \in \mathcal{M}}\sum_{d \in \mathcal{D}} \lambda_{ld}\mu_{ldm} p_{ldm},
\end{align}
where we have $(b)$ due to that the ZF precoder only focuses on the desired signal, and hence the ZF precoding vector is uncorrelated with the SI channel, we have$(c)$ according to the assumption that $\left\|\pmb{f}_{ld}^{\text{ZF}}[m]\right\|_2^2=1$ and $\mathbb{E}\left[\left(\pmb{H}_{ll}^{(i,:)}\right)^H\pmb{H}_{ll}^{(i,:)}\right]=\xi_l^{\text{SI}}\pmb{I}_N$.

Similarly, the third term in the denominator of \eqref{eq:MDD_CF:ulSINR} can be obtained as
\begin{align}
&\sum_{l \in \mathcal{L}}\mathbb{E}\left[\left\|\pmb{w}_{ld}^H[\bar{m}]\pmb{z}^{\text{IAI}}_l\right\|^2\right] \nonumber \\
&=\xi_l^{\text{IAI}}\sum_{l \in \mathcal{L}}\upsilon_{ld\bar{m}}\sum_{l^{\prime} \in \mathcal{L}\backslash \left\{l\right\}}\sum_{m \in \mathcal{M}}\sum_{d \in \mathcal{D}} \frac{\beta_{ll^{\prime}}}{M_{\text{sum}}}\lambda_{l^{\prime}d}\mu_{l^{\prime}dm} p_{l^{\prime}dm},
\end{align}
where $\upsilon_{ld\bar{m}}=\left\|\pmb{w}_{ld}^{\text{ZF}}[\bar{m}]\right\|_2^2$. Consequently, $\text{SINR}_{d,m}$ and $\text{SINR}_{d,\bar{m}}$ can be simplified to the desired form as shown in \eqref{eq:MDDCF:SINRSim}.

\ifCLASSOPTIONcaptionsoff
  \newpage
\fi



%

\bibliographystyle{IEEEtran}
\bibliography{MDDRef}

%








\end{document}